\documentclass[prb,twocolumn,superscriptaddress,aps,nofootinbib]{revtex4-2}
\usepackage{microtype}
\usepackage{graphicx}
\usepackage{amssymb,amsmath,amsthm,mathtools,amscd}
\usepackage{bm}
\usepackage[T1]{fontenc}
\usepackage[utf8]{inputenc}
\usepackage{lmodern}
\usepackage{textcomp}
\usepackage{amstext}
\usepackage[Symbol]{upgreek}
\usepackage{subfigure}
\usepackage{booktabs}
\usepackage{dcolumn}
\usepackage{color}
\usepackage[table]{xcolor}
\usepackage{enumerate}
\usepackage{latexsym}
\usepackage{longtable}
\usepackage{multirow}
\usepackage{makecell}
\usepackage[colorlinks=true,allcolors=blue,urlcolor=blue,citecolor=blue]{hyperref}
\usepackage{booktabs,array,tabularx}
\usepackage{CJKutf8}

\newtheorem{proposition}{Proposition}
\newtheorem{corollary}{Corollary}

\usepackage[normalem]{ulem}


\begin{document}

\begin{CJK*}{UTF8}{gbsn}

\title{A Pseudoscalar Representation Mapping from Parent-Group Vibrational Normal Modes to Symmetry-Adapted Magnetic Structures}

\author{Yachao Liu (刘雅超)}
\email{liuyachao@xaut.edu.cn}
\affiliation{Department of Applied Physics, Xi'an University of Technology, Xi'an 710054, China}
\affiliation{Interdisciplinary Research Center for Semiconductor and Artificial Intelligence, Xi'an University of Technology, Xi'an 710048, China}

\author{Haibo Niu (牛海波)}
\affiliation{Department of Applied Physics, Xi'an Jiaotong University City College, Xi'an 710018, China}

\author{Vei Wang (王伟)}
\email{wangvei@icloud.com}
\affiliation{Department of Electronic Engineering, Xi'an University of Technology, Xi'an 710048, China}
\affiliation{Interdisciplinary Research Center for Semiconductor and Artificial Intelligence, Xi'an University of Technology, Xi'an 710048, China}

\date{\today}

\begin{abstract}
Conventional approaches classify symmetry-allowed magnetic configurations but do not by themselves establish a direct, mode-resolved correspondence with parent-lattice vibrations. Here, we formulate a universal determinant-induced pseudoscalar twist for all 32 crystallographic point groups, establishing an exact representation-to-geometry correspondence between parent vibrations and magnetic order. Within the paramagnetic gray group $G\times\Theta_{\mathcal{T}}$, the spatial twist determines symmetry-defined magnetic geometry, while time-reversal parity independently specifies magnetic character. Each parent phonon irrep $\Gamma$ maps to $\Gamma_{\mathrm{mag}}=\Gamma\otimes\Gamma_{\mathrm{ps}}$, preserving multiplicities and yielding the projection identity $P_{\mathrm{mag},mn}^{(\Gamma\otimes\Gamma_{\mathrm{ps}})}=P_{\mathrm{ph},mn}^{(\Gamma)}$ under the common Cartesian realization. This establishes the \textit{Template Principle}: parent vibrational modes furnish real-space templates whose symmetry-enforced nodal manifolds are inherited exactly. Applied to monolayer $\mathrm{Cd}_2\mathrm{N}_3$, the framework identifies the ferrimagnetic ground state from the parent $A_{2u}$ sector, confirmed by first-principles calculations, alongside cluster magnetic octupoles and antiferromagnetic manifolds. It further provides an \textit{a priori} parent-group criterion for screening symmetry-allowed linear magnetic responses.
\end{abstract}

\keywords{pseudoscalar representation; parent-group symmetry; magnetic representation analysis; symmetry-adapted magnetic structures; lattice-vibrational templates}

\maketitle
\end{CJK*}

% --- Section I: Introduction ---
\section{Introduction}
\label{sec:intro}

A crystalline parent phase possesses a complete set of symmetry-organized lattice-vibrational normal modes before any magnetic or structural phase transition occurs. Magnetic ordering, by contrast, is often discussed in terms of physically realized spin configurations obtained through experiment, energetic searches, or thermodynamic selection, even though the symmetry-allowed magnetic representation space can already be formulated from the parent crystal. This creates a methodological asymmetry: while structural distortions are naturally associated with physically realizable parent-lattice modes, the corresponding real-space realizations of symmetry-adapted magnetic order are usually not directly connected to those pre-existing vibrational degrees of freedom. This raises a fundamental physical question: \textit{does the configuration space of symmetry-adapted magnetic order already reside within the spatial mode structure encoded in the parent crystal?}

The physical origin of this polar--axial distinction lies in the different transformation properties of vector fields: atomic displacements transform as polar vectors, whereas magnetic moments transform as axial vectors under spatial point-group operations. Furthermore, structural vibrations are time-reversal even ($\mathcal{T}\mathbf{u}=\mathbf{u}$), whereas magnetic order parameters are time-reversal odd ($\mathcal{T}\mathbf{M}=-\mathbf{M}$). In the classical macroscopic regime considered here, the full symmetry of the paramagnetic parent phase is described by the gray magnetic parent group $\mathcal{G}_{\mathrm{P}}=G\times\Theta_{\mathcal{T}}$, where $G$ is the unitary spatial crystallographic point group and $\Theta_{\mathcal{T}}=\{E,\mathcal{T}\}$ is the time-reversal group. This direct-product structure permits a separation of spatial transformation character from time-reversal parity: the spatial group $G$ determines the geometric symmetry of a magnetic configuration, while $\Theta_{\mathcal{T}}$ provides an independent time-reversal grading, with magnetic order carrying $\tau_{\mathcal{T}}=-1$. The polar--axial distinction therefore suggests a universal spatial mapping between parent-lattice modes and magnetic configurations rather than an independent search over unconstrained spin patterns.

Existing group-theoretical methodologies---including Bertaut's magnetic representation analysis and isotropy-subgroup theory \cite{Bertaut1962,Bertaut1963,Bertaut1968,Bertaut1971,Bradley2010, Stokes1988_Isotropy}---rigorously determine symmetry-allowed magnetic representation spaces, symmetry-adapted basis functions, and order-parameter isotropy branches from parent groups and Wyckoff positions. The unresolved issue addressed here is therefore not the abstract classification of magnetic representations, but their mode-resolved physical realization within the parent lattice: is there a direct correspondence between a parent-lattice vibrational mode and the spatial configuration of the corresponding symmetry-adapted magnetic order parameter?

In this work, we resolve this issue by formulating a universal determinant-induced pseudoscalar twist that acts involutively on the irreducible-representation landscape of all 32 crystallographic point groups, promoting the pseudoscalar from a vector-level polar--axial parity factor to a structural organizing principle across crystallographic representation spaces. At the single-irrep level, each parent-lattice vibrational sector $\Gamma$ maps to the symmetry-adapted magnetic sector $\Gamma_{\mathrm{mag}}=\Gamma\otimes\Gamma_{\mathrm{ps}}$, extending for the complete crystal representation to $\Gamma_{\mathrm{mag}}^{\mathrm{tot}} = \Gamma_{\mathrm{ph}}^{\mathrm{tot}} \otimes\Gamma_{\mathrm{ps}}$. At the realization level, this twist strictly preserves representation multiplicities ($n_{\mathrm{mag}}^{(\Gamma\otimes\Gamma_{\mathrm{ps}})} = n_{\mathrm{ph}}^{(\Gamma)}$) and, crucially, renders the corresponding polar and axial matrix-element projection operators identical under the common site-resolved Cartesian coefficient space ($P_{\mathrm{mag},mn}^{(\Gamma\otimes\Gamma_{\mathrm{ps}})} = P_{\mathrm{ph},mn}^{(\Gamma)}$). This exact representation-to-geometry correspondence establishes the \textit{Template Principle}: parent-phase vibrational normal modes furnish concrete, symmetry-defined real-space templates for symmetry-adapted magnetic order parameters, with symmetry-enforced nodal constraints strictly preserved. The key conceptual distinction is therefore between the parent-lattice realization of the possibility space and the thermodynamic selection of a physical state: representation theory defines the symmetry-allowed configuration space, while our mapping provides its concrete parent-lattice templates, upon which microscopic energetics determines the realized ground state. This establishes a clear physical progression from \textit{Possibility} to \textit{Selection}, and ultimately to \textit{Response}.

We demonstrate the concrete physical realization of this framework in monolayer $\mathrm{Cd}_2\mathrm{N}_3$. The parent-group vibrational mapping identifies the symmetry channel and real-space configuration of the $A_{2g}$ out-of-plane ferrimagnetic ground state (originating from the parent $A_{2u}$ phonon sector) predicted by first-principles calculations, while revealing $B_{1g}$ and $B_{2g}$ cluster magnetic octupoles alongside an $E_{2g}$ in-plane antiferromagnetic configuration sector directly from parent-lattice kinematics. Beyond magnetic-structure construction, the same spatial-temporal separation principle provides a general parent-group criterion for symmetry-allowed linear magnetic responses. The present work therefore organizes the symmetry-allowed magnetic possibilities prior to energetic selection and provides the algebraic foundation for companion studies of linear magnetoelectricity and piezomagnetism.

% --- Section II: Pseudoscalar Mapping Theory ---
\section{Pseudoscalar Mapping Theory}
\label{sec:pseudoscalar_mapping}

For a given crystalline material, the full paramagnetic parent symmetry at the point-group level is described by the gray magnetic point group $\mathcal{G}_{\mathrm{P}} = G \times \Theta_{\mathcal{T}}$, where $G$ is the unitary crystallographic spatial parent group and $\Theta_{\mathcal{T}} = \{E,\mathcal{T}\}$ is the independent time-reversal group. For the classical macroscopic magnetic structures considered here, spatial transformation properties and time-reversal parity are treated as independent symmetry labels. Accordingly, the present formulation operates on the spatial sector $G$ to establish the universal polar-to-axial geometric correspondence through a pseudoscalar representation $\Gamma_{\mathrm{ps}}$, while the time-reversal-odd magnetic character is assigned independently. Spinorial double-group representations and the full Wigner corepresentation formalism are therefore outside the scope of the present classical treatment.

The core of the framework is the universal determinant-induced pseudoscalar twist acting on the spatial representation landscape. For any crystallographic point group $G$, $\Gamma_{\mathrm{ps}}$ is the canonical one-dimensional sign representation whose character is
\begin{equation}
    \chi_{\mathrm{ps}}(g)=\det(R_g),
\end{equation}
where $R_g\in O(3)$ is the orthogonal matrix representing the spatial operation $g\in G$ in real space. Since every orthogonal transformation satisfies $\det(R_g)=\pm1$, $\Gamma_{\mathrm{ps}}$ is real and self-inverse,
\begin{equation}
    \Gamma_{\mathrm{ps}}\otimes\Gamma_{\mathrm{ps}}=\Gamma_1.
\end{equation}
Physically, $\Gamma_{\mathrm{ps}}$ distinguishes orientation-preserving from orientation-reversing spatial operations. In this work, we promote the pseudoscalar from a vector-level polar--axial parity factor to an involutive twist acting on the entire irreducible-representation landscape of the spatial parent group; its time-reversal character is treated independently in Sec.~\ref{sec:time_reversal}.

Before addressing collective lattice modes, we establish the fundamental polar--axial duality at the single-vector level. Under a point-group operation $g$, a polar vector $\mathbf{u}$, such as an atomic displacement, transforms according to the natural representation $D^{(p)}(g)=R_g$, whereas an axial vector $\mathbf{m}$, such as a magnetic moment, transforms as $D^{(a)}(g)=\det(R_g)R_g$ \cite{Landau8}. In representation-theoretic language, this duality is expressed as
\begin{equation}
    \boxed{
    \Gamma_{\mathrm{axial}}
    =
    \Gamma_{\mathrm{polar}}\otimes\Gamma_{\mathrm{ps}}
    }.
    \label{eq:vec_dual}
\end{equation}
Thus, $\Gamma_{\mathrm{ps}}$ provides the canonical one-dimensional geometric bridge that converts a polar-vector representation into its axial counterpart.

We next extend this vector duality to an individual symmetry sector of the parent lattice dynamics. Any symmetry-adapted phonon sector transforming according to an irreducible representation $\Gamma$ of $G$ maps to a corresponding symmetry-adapted magnetic sector
\begin{equation}
    \boxed{
    \Gamma_{\mathrm{mag}}
    =
    \Gamma\otimes\Gamma_{\mathrm{ps}}
    }.
    \label{eq:irrep_map}
\end{equation}
This one-to-one mapping assigns to every spatial vibrational irreducible representation a unique symmetry-adapted magnetic partner. A particular phonon realization within the $\Gamma$ sector then provides a corresponding real-space magnetic template within the mapped sector $\Gamma_{\mathrm{mag}}$. The complete crystal representation mapping and its consequences for representation multiplicities and real-space projection operators are derived in Sec.~\ref{sec:rep_analysis}.

The action of $\Gamma_{\mathrm{ps}}$ naturally divides the 32 crystallographic point groups into three classes, summarized in Table~\ref{tab:table1}.

\textit{Class I: Chiral.---} In chiral point groups containing only proper rotations, $\det(R_g)=+1$ for all $g\in G$. Consequently, $\Gamma_{\mathrm{ps}}=\Gamma_1$, and the polar and axial sectors are symmetry-equivalent within $G$.

\textit{Class II: Centrosymmetric.---} In centrosymmetric point groups containing spatial inversion, $\chi_{\mathrm{ps}}(i)=-1$, so tensoring with $\Gamma_{\mathrm{ps}}$ reverses inversion parity, $g\leftrightarrow u$.

\textit{Class III: Achiral non-centrosymmetric.---} In non-centrosymmetric point groups containing improper operations, $\Gamma_{\mathrm{ps}}$ is nontrivial and records the determinant signatures of mirrors and roto-inversions. The resulting pseudoscalar twist changes the irrep labels of the symmetry-adapted magnetic sectors and thereby distinguishes spatial symmetry channels relevant to noncollinear and antiferromagnetic configurations.

Crucially, because all crystallographic spatial operations are orthogonal, $\chi_{\mathrm{ps}}(g)\in\{+1,-1\}$, and the mapping is involutive:
\begin{equation}
    \Gamma\otimes\Gamma_{\mathrm{ps}}\otimes\Gamma_{\mathrm{ps}}
    =
    \Gamma.
\end{equation}
Thus, applying the pseudoscalar twist twice restores the original representation. This algebraic involution provides the basis for the exact equivalence of the polar and axial projection operators established in Sec.~\ref{sec:rep_analysis}. In categorical language, this involutive structure is realized as an autoequivalence of $\mathrm{Rep}(G)$ (Appendix~\ref{subsec:rep_category}).

\begin{table*}[t]
\centering
\small
\caption{Pseudoscalar representations $\Gamma_{\mathrm{ps}}$ for all 32 crystallographic point groups, classified by the determinant of symmetry operations.
\label{tab:table1}}
\renewcommand{\arraystretch}{1.12}
\setlength{\tabcolsep}{6pt}
\begin{tabular}{llll}
\toprule
\textbf{Category} & \textbf{International Symbol (H--M)} &
\textbf{Schoenflies Symbol} & $\boldsymbol{\Gamma_{\mathrm{ps}}}$ \\
\midrule
\multirow{2}{*}{\textbf{1. Chiral Point Groups}}
& $1, 2, 3, 4, 6, 222, 23$
& $C_1, C_2, C_3, C_4, C_6, D_2, T$
& $A$ \\
& $422, 32, 622, 432$
& $D_4, D_3, D_6, O$
& $A_1$ \\
\midrule
\multirow{2}{*}{\textbf{2. Centrosymmetric Point Groups}}
& $\bar{1}, 2/m, 4/m, 6/m, mmm, m\bar{3}, \bar{3}$
& $C_i, C_{2h}, C_{4h}, C_{6h}, D_{2h}, T_h, S_6$
& $A_u$ \\
& $\bar{3}m, 4/mmm, 6/mmm, m\bar{3}m$
& $D_{3d}, D_{4h}, D_{6h}, O_h$
& $A_{1u}$ \\
\midrule
\multirow{4}{*}{\textbf{3. Achiral Point Groups}}
& $m, \bar{6}$
& $C_h(C_s), C_{3h}$
& $A^{\prime\prime}$ \\
& $\bar{6}m2$
& $D_{3h}$
& $A_1^{\prime\prime}$ \\
& $2mm, 3m, 4mm, 6mm, \bar{4}3m$
& $C_{2v}, C_{3v}, C_{4v}, C_{6v}, T_d$
& $A_2$ \\
& $\bar{4}, \bar{4}2m$
& $S_4, D_{2d}$
& $B \,/\, B_1$ \\
\bottomrule
\end{tabular}
\vspace{0.5em}
\begin{minipage}{0.95\textwidth}
\small
\textbf{Note:}
For all crystallographic point groups, $\Gamma_{\mathrm{ps}}$ is the one-dimensional sign representation associated with the orthogonal spatial action, $\chi_{\mathrm{ps}}(g)=\det(R_g)$.
\end{minipage}
\end{table*}

% --- Section III: Representation Analysis and Template Principle ---
\section{Representation Analysis and Template Principle}
\label{sec:rep_analysis}

Having established the pseudoscalar mapping for single vectors and individual irreducible representations, we now establish its consequences for the full parent-lattice representation: the total magnetic representation mapping, the exact preservation of irreducible-representation multiplicities, and the equivalence of the corresponding symmetry-projected configuration subspaces. This formulation connects the pseudoscalar mapping to the classical representation analysis of magnetic structures pioneered by Bertaut \cite{Bertaut1968,Bertaut1971}. Here, magnetic moments are treated as classical axial vectors, a framework well suited to symmetry-based descriptions of macroscopic magnetic structures, while quantum-fluctuation-dominated regimes require additional microscopic treatments.

\subsection{Total Representation Mapping and Multiplicity Preservation}
\label{subsec:tot_rep_mapping}

For a crystal lattice containing $N$ atoms, the total vibrational representation $\Gamma_{\mathrm{ph}}^{\mathrm{tot}}$ is constructed from the atomic site-permutation representation $\Gamma_{\mathrm{eq}}$ and the polar-vector representation $\Gamma_{\mathrm{polar}}$ \cite{Dresselhaus2008}:
\begin{equation}
    \Gamma_{\mathrm{ph}}^{\mathrm{tot}}
    =
    \Gamma_{\mathrm{eq}}\otimes\Gamma_{\mathrm{polar}}
    =
    \bigoplus_{\Gamma}
    n_{\mathrm{ph}}^{(\Gamma)}\,\Gamma,
    \label{eq:tot_ph}
\end{equation}
where $\Gamma$ denotes an irreducible representation of the parent group $G$. Its multiplicity in the parent-lattice vibrational space is
\begin{equation}
    n_{\mathrm{ph}}^{(\Gamma)}
    =
    \frac{1}{|G|}
    \sum_{g\in G}
    \chi^{(\Gamma)}(g)^*
    \chi_{\mathrm{ph}}^{\mathrm{tot}}(g).
    \label{eq:n_ph}
\end{equation}

In parallel, the total magnetic representation $\Gamma_{\mathrm{mag}}^{\mathrm{tot}}$ describes the collective transformation of localized magnetic moments over the same crystallographic sites and is constructed from the site-permutation representation and the axial-vector representation \cite{Bertaut1968}:
\begin{equation}
    \Gamma_{\mathrm{mag}}^{\mathrm{tot}}
    =
    \Gamma_{\mathrm{eq}}\otimes\Gamma_{\mathrm{axial}}.
    \label{eq:tot_mag_def}
\end{equation}
Using the polar--axial relation $\Gamma_{\mathrm{axial}} = \Gamma_{\mathrm{polar}}\otimes\Gamma_{\mathrm{ps}}$ and the associativity of tensor products, we obtain the global mapping theorem for the complete parent-lattice configuration space:
\begin{equation}
    \boxed{
    \Gamma_{\mathrm{mag}}^{\mathrm{tot}}
     =
    \Gamma_{\mathrm{ph}}^{\mathrm{tot}}\otimes\Gamma_{\mathrm{ps}}
    }.
    \label{eq:tot_mag_map}
\end{equation}
Consequently, $\chi_{\mathrm{mag}}^{\mathrm{tot}}(g) = \chi_{\mathrm{ph}}^{\mathrm{tot}}(g) \chi_{\mathrm{ps}}(g)$.

Decomposing the magnetic representation into irreducible representations,
\begin{equation}
    \Gamma_{\mathrm{mag}}^{\mathrm{tot}}
    =
    \bigoplus_{\Gamma_{\mathrm{mag}}}
    n_{\mathrm{mag}}^{(\Gamma_{\mathrm{mag}})}
    \,\Gamma_{\mathrm{mag}},
\end{equation}
the multiplicity of the mapped irrep $\Gamma_{\mathrm{mag}}=\Gamma\otimes\Gamma_{\mathrm{ps}}$ is
\begin{equation}
\begin{aligned}
    n_{\mathrm{mag}}^{(\Gamma\otimes\Gamma_{\mathrm{ps}})}
    &=
    \frac{1}{|G|}
    \sum_{g\in G}
    \left[
    \chi^{(\Gamma)}(g)
    \chi_{\mathrm{ps}}(g)
    \right]^*
    \left[
    \chi_{\mathrm{ph}}^{\mathrm{tot}}(g)
    \chi_{\mathrm{ps}}(g)
    \right].
\end{aligned}
\end{equation}
Because $\chi_{\mathrm{ps}}(g)=\det(R_g)=\pm1$ is real and satisfies $[\chi_{\mathrm{ps}}(g)]^2=1$ for every $g\in G$, the pseudoscalar factors cancel identically. The result reduces exactly to Eq.~\eqref{eq:n_ph}, giving
\begin{equation}
    \boxed{
    n_{\mathrm{mag}}^{(\Gamma\otimes\Gamma_{\mathrm{ps}})}
    =
    n_{\mathrm{ph}}^{(\Gamma)}
    }.
    \label{eq:template1}
\end{equation}

Thus the pseudoscalar twist changes the spatial representation label according to $\Gamma\mapsto\Gamma\otimes\Gamma_{\mathrm{ps}}$ while strictly preserving the occurrence multiplicity. Since the dimension of an isotypic sector is $n_\Gamma d_\Gamma$, the corresponding polar and axial sectors also have identical dimensions.

\subsection{Projection Operator Equivalence and the Template Principle}
\label{subsec:proj_op_equiv}

Multiplicity preservation establishes equality of sector dimensions, but the central geometric statement of the framework follows from the corresponding projection operators. Consider first a polar displacement field. The matrix-element projector onto the $mn$ component of the irrep $\Gamma$ is
\begin{equation}
    P_{\mathrm{ph},mn}^{(\Gamma)}
    =
    \frac{d_\Gamma}{|G|}
    \sum_{g\in G}
    D_{mn}^{(\Gamma)}(g)^*
    \hat D_{\mathrm{pol}}(g),
\end{equation}
where $d_\Gamma=\dim\Gamma$, $D_{mn}^{(\Gamma)}(g)$ is the irrep matrix element, and $\hat D_{\mathrm{pol}}(g)$ is the spatial action on the site-resolved polar coefficient space.

For the corresponding axial magnetic realization $\Gamma_{\mathrm{mag}}=\Gamma\otimes\Gamma_{\mathrm{ps}}$, $D_{mn}^{(\Gamma\otimes\Gamma_{\mathrm{ps}})}(g) = D_{mn}^{(\Gamma)}(g)\chi_{\mathrm{ps}}(g)$, while $\hat D_{\mathrm{ax}}(g) = \chi_{\mathrm{ps}}(g)\hat D_{\mathrm{pol}}(g)$. Hence,
\begin{align}
P_{\mathrm{mag},mn}^{(\Gamma\otimes\Gamma_{\mathrm{ps}})}
&=
\frac{d_\Gamma}{|G|}
\sum_{g\in G}
\left[
D_{mn}^{(\Gamma)}(g)\chi_{\mathrm{ps}}(g)
\right]^*
\left[
\chi_{\mathrm{ps}}(g)\hat D_{\mathrm{pol}}(g)
\right]
\nonumber\\
&=
\frac{d_\Gamma}{|G|}
\sum_{g\in G}
D_{mn}^{(\Gamma)}(g)^*
[\chi_{\mathrm{ps}}(g)]^2
\hat D_{\mathrm{pol}}(g).
\end{align}
Since $[\chi_{\mathrm{ps}}(g)]^2=1$, the pseudoscalar factors cancel element by element, giving the exact operator identity:
\begin{equation}
    \boxed{
    P_{\mathrm{mag},mn}^{(\Gamma\otimes\Gamma_{\mathrm{ps}})}
    =
    P_{\mathrm{ph},mn}^{(\Gamma)}
    }.
    \label{eq:template2}
\end{equation}
Here and below, the equality is understood under the natural identification of the common site-resolved Cartesian coefficient space. Thus the two operators have identical action on the underlying Cartesian coefficient vectors, although they carry different $G$-representation labels.

Summing the diagonal matrix-element projectors, $\sum_m P_{mm}^{(\Gamma)}$, gives the corresponding character projector:
\begin{equation}
    \boxed{
    P_{\mathrm{mag}}^{(\Gamma\otimes\Gamma_{\mathrm{ps}})}
    =
    P_{\mathrm{ph}}^{(\Gamma)}
    }.
    \label{eq:char_proj}
\end{equation}

For a one-dimensional irrep occurring with unit multiplicity in the relevant magnetic configuration space, $d_\Gamma=n_\Gamma=1$, the corresponding symmetry sector is one-dimensional. Since $D_{11}^{(\Gamma)}=\chi^{(\Gamma)}$, the character projector is then sufficient to extract the unique symmetry-adapted basis mode from any seed vector with nonzero projection onto that sector. The resulting normalized vector is fixed uniquely up to an overall amplitude and sign.

More generally, the operator identity in Eq.~\eqref{eq:char_proj} establishes equality of the corresponding projected subspaces:
\begin{equation}
    \boxed{
    \mathrm{Im}
    \left(
    P_{\mathrm{mag}}^{(\Gamma\otimes\Gamma_{\mathrm{ps}})}
    \right)
    =
    \mathrm{Im}
    \left(
    P_{\mathrm{ph}}^{(\Gamma)}
    \right)
    }.
    \label{eq:isotypic_iso}
\end{equation}
The equality here refers to the common Cartesian realization specified above; abstractly, the corresponding isotypic representation spaces are naturally isomorphic, $\mathrm{Im}(P_{\mathrm{mag}}^{(\Gamma\otimes\Gamma_{\mathrm{ps}})}) \cong \mathrm{Im}(P_{\mathrm{ph}}^{(\Gamma)})$. Equation~\eqref{eq:isotypic_iso} gives the geometric content of the \textit{Template Principle}: parent-phase vibrational normal modes and the corresponding symmetry-adapted magnetic configurations occupy the same projected real-space coefficient subspace under the natural polar--axial identification.

Together with the site-stabilizer covariance established in Appendix~\ref{subsec:projection_nodal_app}, this equivalence implies that all symmetry-enforced nodal constraints---including real-space positions and local Cartesian components forced to vanish by local stabilizer symmetry---are strictly preserved under the pseudoscalar mapping. Accidental, material-dependent zeros are not constrained by this statement. Thus, parent-phase vibrational normal modes provide faithful real-space templates for symmetry-adapted magnetic order parameters.

\subsection{Magnetic-Sublattice Representation and Phonon-Derived Templates}
\label{subsec:sublattice}

In practical crystalline materials, magnetic moments typically reside on a specific subset of Wyckoff positions (the magnetic sublattice $\mathcal{S}_{\mathrm{mag}}$). To treat site selectivity rigorously, our framework decouples the formal representation-theoretic classification from the real-space physical extraction of vibrational templates: \textit{the sublattice-restricted representation determines the complete symmetry-allowed magnetic configuration space, whereas the magnetic-sublattice projection of first-principles phonon eigenvectors provides concrete crystallographic realizations within that space.}

\textit{1. Sublattice-Restricted Representation Space.---} At the formal group-theoretic level, provided that $\mathcal{S}_{\mathrm{mag}}$ is an invariant union of parent-group site orbits, the complete magnetic configuration space allowed on the magnetic sublattice is obtained by restricting the site-permutation representation $\Gamma_{\mathrm{eq}}$ to $\mathcal{S}_{\mathrm{mag}}$:
\begin{equation}
    \Gamma_{\mathrm{mag}}^{\text{(sub)}} 
    = \left( \Gamma_{\mathrm{eq}}\big\vert_{\mathrm{mag}} \right) \otimes \Gamma_{\mathrm{axial}}
    = \left( \Gamma_{\mathrm{eq}}\big\vert_{\mathrm{mag}} \right) \otimes \Gamma_{\mathrm{polar}} \otimes \Gamma_{\mathrm{ps}},
    \label{eq:sub_mag}
\end{equation}
where $\Gamma_{\mathrm{eq}}\big\vert_{\mathrm{mag}}$ spans the permutation space of the magnetic sublattice only. Equation~\eqref{eq:sub_mag} therefore defines the complete symmetry-allowed magnetic configuration space on the magnetic sublattice, with irreducible-representation content
\begin{equation}
\Gamma_{\mathrm{mag}}^{\mathrm{(sub)}} =
\bigoplus_{\Gamma} n_{\Gamma}^{(\mathrm{sub})}
\left(\Gamma\otimes\Gamma_{\mathrm{ps}}\right).
 \label{eq:sub_mag_decomp} 
\end{equation}
The corresponding symmetry-adapted magnetic basis modes can then be constructed directly within this subspace by applying the matrix-element projection operators $P_{\mathrm{mag},mn}^{(\Gamma\otimes\Gamma_{\mathrm{ps}})}$ to suitable seed configurations on the magnetic sublattice. This representation-theoretic route is formally complete and does not require a lattice-dynamical calculation.

\textit{2. Sublattice Restriction and Phonon-Derived Magnetic Templates.---} A complementary physical route is provided by the parent-lattice vibrational eigenmodes: rather than constructing the magnetic basis abstractly from the restricted magnetic representation, one can extract concrete real-space magnetic templates by restricting full-cell phonon eigenvectors to the magnetic sublattice. When concrete lattice dynamical calculations are available, physical templates are extracted by applying a site-restriction map $Q_{\mathrm{mag}}: V_{\mathrm{full}} \to V_{\mathrm{mag}}$ that retains displacement vectors on $\mathcal{S}_{\mathrm{mag}}$ while omitting non-magnetic ligand coordinates, defining the restricted polar pattern:
\begin{equation}
    \lvert \boldsymbol{\epsilon}_{\nu}^{\mathrm{sub}} \rangle = Q_{\mathrm{mag}} \lvert \boldsymbol{\epsilon}_{\nu}^{\mathrm{ph}} \rangle.
\end{equation}
The resulting restricted polar pattern $\lvert \boldsymbol{\epsilon}_\nu^{\mathrm{sub}} \rangle$ is directly reinterpreted as an axial magnetic template $\lvert \widetilde{\mathbf{M}}_\nu \rangle \equiv Q_{\mathrm{mag}} \lvert \boldsymbol{\epsilon}_\nu^{\mathrm{ph}} \rangle$, with the polar-to-axial distinction encoded entirely in the transformed group action $D_{\mathrm{ax}}(g) = \chi_{\mathrm{ps}}(g) D_{\mathrm{pol}}(g)$ under the common underlying site-resolved Cartesian coefficient space. Provided that the magnetic sublattice consists of complete parent-group orbits (i.e., is an invariant site set satisfying $g\mathcal{S}_{\mathrm{mag}} = \mathcal{S}_{\mathrm{mag}}$ for all $g \in G$), the restriction map satisfies the intertwining relation $Q_{\mathrm{mag}} D_{\mathrm{full}}(g) = D_{\mathrm{mag}}(g) Q_{\mathrm{mag}}$. Consequently, provided $Q_{\mathrm{mag}} \lvert \psi_{\Gamma} \rangle \neq \mathbf{0}$, the projected vector strictly preserves its parent-group irrep, $Q_{\mathrm{mag}} \lvert \psi_{\Gamma} \rangle \in V_{\mathrm{mag}}^{(\Gamma)}$.

\textit{3. Conceptual Hierarchy, Multiplicity Collapse, and Optical Constraints.---} Crucially, the operators $Q_{\mathrm{mag}}$, $P_{\mathrm{mag}, mn}^{(\Gamma)}$, and the pseudoscalar twist $\mathcal{F}_{\mathrm{ps}}$ act at three distinct conceptual levels: site restriction, irrep projection, and representation twisting, respectively. Specifically, $Q_{\mathrm{mag}}$ is a spatial site-restriction map and does not belong to the representation twist $\mathcal{F}_{\mathrm{ps}}$. 

While $Q_{\mathrm{mag}}: V_{\mathrm{full}} \to V_{\mathrm{mag}}$ is the canonical coordinate projection from the full ambient coordinate space $V_{\mathrm{full}}$ onto $V_{\mathrm{mag}}$, it is not injective. When restricted to a specific phonon-mode subspace $W_{\mathrm{ph}} \subset V_{\mathrm{full}}$, the restricted map $Q_{\mathrm{mag}}\big\vert_{W_{\mathrm{ph}}}$ is generally neither injective nor surjective onto $V_{\mathrm{mag}}$. Within hybridized multi-dimensional or multi-copy phonon sectors ($n_\Gamma > 1$), distinct full-cell vibrational eigenvectors can project onto collinear or identical magnetic templates on $\mathcal{S}_{\mathrm{mag}}$ (loss of multiplicity under site restriction).

Furthermore, full-cell $\Gamma$-point optical phonons are constrained to lie in the subspace orthogonal to rigid acoustic translations, satisfying the mass-weighted center-of-mass condition $\sum_{i \in \text{all}} m_i \mathbf{u}_i = \mathbf{0}$. This global lattice constraint is generally not preserved under sublattice restriction ($Q_{\mathrm{mag}} V_{\mathrm{opt}}^{\mathrm{full}} \subseteq V_{\mathrm{mag}}$, but $Q_{\mathrm{mag}} V_{\mathrm{opt}}^{\mathrm{full}} \neq V_{\mathrm{mag}}$ in general). Therefore, $\Gamma_{\mathrm{mag}}^{\mathrm{(sub)}}$ is a magnetic configuration space rather than an optical phonon space. Removing non-magnetic sites can also alter the conventional physical labeling of a projected pattern: a full-cell optical pattern with opposite displacements on two sublattices may project onto a single magnetic sublattice and appear as a uniform ferromagnetic-like configuration. Irrep symmetry is strictly preserved, but magnetic-order nomenclature (FM/FiM/AFM) is not invariant under site restriction. In this extraction procedure, non-magnetic ligand coordinates are omitted solely to isolate the local spin-moment geometry, without implying that ligand-mediated superexchange or dynamic spin-lattice pathways are negligible.

\subsection{Symmetry-Adapted Magnetic Basis and Order-Parameter Selection}
\label{subsec:decoupling}

To distinguish lattice-dynamical realizations from symmetry-adapted magnetic configurations, our framework clearly separates two distinct classes of mathematical objects:
\begin{enumerate}[(i)]
    \item \textbf{Phonon dynamical eigenvectors $\lvert \boldsymbol{\epsilon}_\nu^{\mathrm{ph}} \rangle$:} Genuine dynamical eigenvectors obtained by diagonalizing the parent-phase lattice dynamical matrix $D_{\mathrm{ph}}\lvert \boldsymbol{\epsilon}_\nu^{\mathrm{ph}} \rangle = \omega_\nu^2 \lvert \boldsymbol{\epsilon}_\nu^{\mathrm{ph}} \rangle$. For an $N$-atom parent unit cell, the $\Gamma$-point displacement space has $3N$ normal-mode eigenvectors, which are organized into symmetry sectors according to the irreducible representations of $G$. Multidimensional irreps correspond to symmetry-related degenerate eigenvectors within a given phonon sector. Within a degenerate phonon eigenspace, individual eigenvectors are basis-dependent; the symmetry-defined representation sector, rather than a particular choice of basis within that degenerate subspace, is invariant.
    \item \textbf{Symmetry-adapted magnetic basis modes $\lvert \phi_a^{(\Gamma_{\mathrm{mag}})} \rangle$:} Pure basis configurations constructed from the magnetic representation $\Gamma_{\mathrm{mag}}^{\mathrm{(sub)}}$ via projection operators $P_{\mathrm{mag}, mn}^{(\Gamma)}$ spanning the corresponding symmetry sector $V_{\mathrm{mag}}^{(\Gamma_{\mathrm{mag}})}$. These modes represent symmetry-defined basis vectors of the magnetic configuration space rather than quantum eigenstates of a dynamic Hamiltonian.
\end{enumerate}

Crucially, for a nonzero order parameter belonging to a non-degenerate one-dimensional representation occurring with unit multiplicity ($d_\Gamma = 1, n_\Gamma = 1$) on the magnetic sublattice, the representation subspace is strictly one-dimensional ($\dim V_\Gamma = n_\Gamma d_\Gamma = 1$). Consequently, spatial symmetry completely and uniquely fixes the relative site-resolved spin pattern up to an overall amplitude scale and global sign, leaving the symmetry-adapted basis mode $\lvert \phi^{(\Gamma_{\mathrm{mag}})} \rangle$ as the unique allowed geometric template.

For multi-dimensional irreps ($d_\Gamma > 1$) and/or multiple occurrences of the same irrep ($n_\Gamma > 1$), phonon-derived magnetic templates $\lvert \widetilde{\mathbf{M}}_\nu \rangle = Q_{\mathrm{mag}} \lvert \boldsymbol{\epsilon}_\nu^{\mathrm{ph}} \rangle$ provide physically resolved vectors (seeds) within the symmetry-allowed subspace, while the complete symmetry-adapted basis $\{\lvert \phi_a^{(\Gamma_{\mathrm{mag}})} \rangle\}_{a=1}^{n_\Gamma d_\Gamma}$ is constructed independently from the formal representation space $\Gamma_{\mathrm{mag}}^{\mathrm{(sub)}}$. Within multidimensional or multiply occurring sectors, phonon eigenvectors provide concrete physical realizations within the symmetry-allowed magnetic representation space, but symmetry alone does not determine the eventual order-parameter direction or energetic preference among their linear combinations.

In general, an arbitrary magnetic configuration within an irrep sector is expressed as a linear superposition:
\begin{equation}
    \lvert \mathbf{M} \rangle = \sum_a \eta_a \lvert \phi_a^{(\Gamma_{\mathrm{mag}})} \rangle,
\end{equation}
where $\boldsymbol{\eta} = (\eta_1, \eta_2, \dots)$ specifies the order-parameter direction (OPD). Parent-group representation theory determines the complete symmetry-allowed magnetic configuration space (\textit{Possibility}); parent-phase phonon eigenvectors provide concrete real-space template realizations within that space; magnetic energetics and thermodynamic free-energy minimization subsequently determine which order-parameter direction and state are realized (\textit{Selection}).

\subsection{Algebraic Protocol for Tensor Decomposition}
\hypertarget{subsec:tensor_decomposition}{} 
\label{subsec:tensor_decomposition}

Beyond vector fields and vibrational modes, the pseudoscalar mapping provides a compact algebraic bookkeeping scheme for parent-group response screening, decomposing macroscopic physical response tensors into polar and axial building blocks. An axial vector transforms as $\Gamma_{\mathrm{ax}} = \Gamma_{\mathrm{ps}} \otimes \Gamma_{\mathrm{pol}}$, where $\Gamma_{\mathrm{pol}}$ is the polar-vector representation. Consequently, a rank-$n$ tensor $T$ containing $n_a$ axial indices and $n_p = n - n_a$ polar indices transforms at the representation level as
\begin{equation}
\Gamma_T
=
\bigotimes_{i=1}^{n_p}\Gamma_{\mathrm{pol}}
\otimes
\bigotimes_{j=1}^{n_a}\left(\Gamma_{\mathrm{ps}}\otimes\Gamma_{\mathrm{pol}}\right)
=
\Gamma_{\mathrm{pol}}^{\otimes n}\otimes \Gamma_{\mathrm{ps}}^{\otimes n_a}.
\end{equation}
Since $\Gamma_{\mathrm{ps}}$ is one-dimensional and involutive ($\Gamma_{\mathrm{ps}}^{\otimes 2} = \Gamma_1$), the pseudoscalar factor reduces to a $\mathbb{Z}_2$ parity label dictated by $n_a \bmod 2$:
\begin{equation}
\Gamma_{\mathrm{ps}}^{\otimes n_a} = 
\begin{cases}
\Gamma_1, & n_a \text{ even}, \\
\Gamma_{\mathrm{ps}}, & n_a \text{ odd}.
\end{cases}
\end{equation}
Thus, tensors with an even number of axial indices transform as ordinary polar tensors, whereas those with an odd number transform as pseudotensors. This representation-theoretic bookkeeping isolates the intrinsic parity structure of physical tensors across arbitrary ranks, supplying the exact selection rules for response screening in Sec.~\ref{sec:discussion}.

% --- Section IV: Time-Reversal Symmetry and Magnetic Subgroups ---
\section{Time-Reversal Symmetry and Magnetic Subgroups}
\label{sec:time_reversal}

\subsection{Time-Reversal Parity, Gray Parent Group, and Spatial-Temporal Separation Principle}
\label{subsec:tr_parity}

In the high-temperature paramagnetic parent phase considered here, the full macroscopic symmetry is formally described by the gray magnetic parent group $\mathcal{G}_{\mathrm{P}} = G \times \Theta_{\mathcal{T}} = G \cup \mathcal{T}G$, where $G$ is the unitary spatial crystallographic point group and $\Theta_{\mathcal{T}} = \{E, \mathcal{T}\}$ is the independent time-reversal group. A fundamental conceptual cornerstone of our framework is the \textit{spatial-temporal separation principle} rooted in this direct-product structure: \textbf{the spatial pseudoscalar twist strictly governs the spatial geometric configuration of the magnetic order, while time-reversal parity ($\tau_{\mathcal{T}} = -1$) independently specifies its magnetic character.}

Physically, a polar displacement field $\mathbf{u}(\mathbf{r})$ and an axial magnetic structure $\mathbf{M}(\mathbf{r})$ differ both geometrically and temporally. For any classical physical quantity or field $X$, the transformation under time reversal is formally defined by
\begin{equation}
    \mathcal{T}X = \tau_{\mathcal{T}}(X)X, \quad \tau_{\mathcal{T}}(X) \in \{+1, -1\},
\end{equation}
where $\tau_{\mathcal{T}}(X) = +1$ and $\tau_{\mathcal{T}}(X) = -1$ explicitly denote time-reversal even and time-reversal odd parities, respectively. While $\mathbf{u}(\mathbf{r})$ is time-even ($\tau_{\mathcal{T}}(\mathbf{u})=+1$), $\mathbf{M}(\mathbf{r})$ is time-odd ($\tau_{\mathcal{T}}(\mathbf{M})=-1$). Thus, within the present classical setting, a magnetic order parameter is characterized by the decoupled symmetry labels $(\Gamma_{\mathrm{spatial}}, \tau_{\mathcal{T}}) = (\Gamma \otimes \Gamma_{\mathrm{ps}}, -1)$. Within this direct-product decomposition $\mathcal{G}_{\mathrm{P}} = G \times \Theta_{\mathcal{T}}$, our mapping $\Gamma_{\mathrm{mag}}=\Gamma\otimes\Gamma_{\mathrm{ps}}$ operates exclusively within the spatial crystallographic parent group $G$ to execute the universal polar-to-axial geometric conversion, fixing the spatial geometric template. The time-reversal oddness ($\tau_{\mathcal{T}} = -1$) is applied as an independent physical label that defines the magnetic nature of the order parameter.

This separation enables a rigorous, constructive definition of a \textit{symmetry-matched magnetic structure}: it is a time-reversal-odd order parameter ($\tau_{\mathcal{T}} = -1$) whose spatial transformation law is compatible with, and directly constructed from, an irreducible representation of the spatial parent crystallographic group $G$. At the full gray-group level, this combined spatial and temporal transformation can be embedded consistently, in the real classical setting considered here, via the factorized time-odd pseudoscalar twist $\widetilde{\Gamma}_{\mathrm{ps}}^- = \Gamma_{\mathrm{ps}} \boxtimes \mathbf{1}^-$. The constructive template derivations below are formulated within the spatial representation category; the corresponding gray-group corepresentation lift is established in Appendix~\ref{subsec:corep_compat}. By explicitly restricting our operational scope to the classical macroscopic limit without invoking spinorial double-group representations or the full quantum-mechanical Wigner corepresentation formalism, this spatial-temporal separation principle ensures that parent-group phonon irreps serve as rigorous, energy-independent spatial templates for emergent magnetic phases.

\subsection{Emergent Magnetic Groups for One-Dimensional Irreps}
\label{subsec:emergent_mpg}

A major practical advantage of this template decoupling is that the emergent magnetic point group (Shubnikov group) $M_{\mathrm{sub}}$ of the ordered phase can be deduced directly from the spatial template. In the Landau theory of continuous phase transitions \cite{Landau_StatPhys1}, the unitary subgroup $H \subset M_{\mathrm{sub}}$ corresponds to the classical isotropy subgroup preserving the condensed order parameter, while the full magnetic point group $M_{\mathrm{sub}}$ incorporates antiunitary elements dictated by time-reversal parity.

For a magnetic structure $\mathbf{M}(\mathbf{r})$ transforming according to a one-dimensional irrep $\Gamma_{\mathrm{mag}} = \Gamma \otimes \Gamma_{\mathrm{ps}}$, any spatial symmetry operation $R\in G$ acts as
\begin{equation}
    R\mathbf{M}(\mathbf{r})=\lambda_R\mathbf{M}(\mathbf{r}), \quad \lambda_R=\chi^{(\Gamma_{\mathrm{mag}})}(R)=\pm 1.
\end{equation}
The emergent magnetic subgroup $M_{\mathrm{sub}}$ is constructed via two canonical rules:
\begin{enumerate}[(i)]
    \item \textbf{Unitary elements:} If a spatial operation $R$ preserves the magnetic template ($\lambda_R=+1$), $R$ remains a unitary symmetry element of $M_{\mathrm{sub}}$.
    \item \textbf{Antiunitary elements:} If $R$ reverses the magnetic configuration ($\lambda_R=-1$), pure spatial symmetry is broken. However, because $\mathcal{T}\mathbf{M}=-\mathbf{M}$, the combined operation $\mathcal{T}R$ restores the structure and enters $M_{\mathrm{sub}}$ as an antiunitary element.
\end{enumerate}

Consequently, for one-dimensional representations, the emergent magnetic point group is explicitly determined by
\begin{equation}
    M_{\mathrm{sub}}=\{R\in G\mid \lambda_R=+1\}\cup\{\mathcal{T}R\mid R\in G,\,\lambda_R=-1\}.
    \label{eq:M_sub}
\end{equation}
Eq.~\eqref{eq:M_sub} generates emergent magnetic point groups that are mathematically consistent with the magnetic subgroups derived from general order-parameter symmetry breaking \cite{Landau_StatPhys1, Michel1980_RMP}.

\subsection{Branch Selection, Order-Parameter Directions, and Epikernels}
\label{subsec:branch_selection}

While non-degenerate one-dimensional irreps with unit multiplicity ($d_\Gamma=1, n_\Gamma=1$) uniquely fix the relative site-resolved spin pattern up to an overall amplitude scale and global sign, a distinct complexity arises for multi-dimensional or multi-copy representation sectors. In this case, the symmetry-adapted basis modes $\mathbf{M}_a(\mathbf{r})$ span a multi-dimensional subspace. The actual condensed magnetic order parameter is a linear combination $\mathbf{M}(\mathbf{r}) = \sum_a \eta_a \mathbf{M}_a(\mathbf{r})$, where $\boldsymbol{\eta} = (\eta_1, \eta_2, \dots)$ defines the order-parameter direction (OPD). Higher-order thermodynamic invariants \cite{Toledano1987} and magnetic anisotropies pin a specific OPD---a physical process known as branch selection.

Once a specific OPD $\boldsymbol{\eta}$ is selected inside the degenerate representation space, the corresponding magnetic subgroup is obtained by identifying spatial operations that leave $\boldsymbol{\eta}$ invariant or reverse it. Define the unitary stabilizer subgroup $H_{\boldsymbol{\eta}}=\{g\in G\mid D(g)\boldsymbol{\eta}=\boldsymbol{\eta}\}$ and the time-reversal-compensating reversing set $A_{\boldsymbol{\eta}}=\{g\in G\mid D(g)\boldsymbol{\eta}=-\boldsymbol{\eta}\}$. The full Shubnikov magnetic point group associated with this specific OPD is then
\begin{equation}
    M_{\mathrm{sub}}(\boldsymbol{\eta})=H_{\boldsymbol{\eta}}\cup\mathcal{T}A_{\boldsymbol{\eta}}.
    \label{eq:M_sub_general}
\end{equation}
For multidimensional order parameters, the remaining elements $g \in G \setminus (H_{\boldsymbol{\eta}} \cup A_{\boldsymbol{\eta}})$ map $\boldsymbol{\eta}$ to symmetry-distinct order-parameter directions (corresponding to different magnetic domains) and therefore do not belong to the magnetic point group of the single domain.

In the language of group-theoretical phase transitions, special high-symmetry OPDs (such as Paths B1 and B2 for $E$-type modes) correspond to \textit{maximal isotropy subgroups} (termed epikernels in the ISOTROPY formalism) \cite{Michel1980_RMP, Stokes1988_Isotropy}. It is essential to clarify the conceptual distinction: basis vectors span the representation subspace, order-parameter directions select specific orientations within that subspace, and epikernels are the maximal residual symmetry subgroups associated with high-symmetry OPDs.

While classical group-theoretical enumeration schemes determine these epikernels purely as abstract algebraic sublattices, our parent-group mapping provides physically resolved real-space realizations of the representation spaces underlying epikernel selection. This four-stage workflow---representation transport ($\Gamma \otimes \Gamma_{\mathrm{ps}} \to \Gamma_{\mathrm{mag}}$), real-space template projection ($P_{\mathrm{mag}, mn}^{(\Gamma \otimes \Gamma_{\mathrm{ps}})} = P_{\mathrm{ph}, mn}^{(\Gamma)}$), physical template extraction ($Q_{\mathrm{mag}}\boldsymbol{\epsilon}_\nu^{\mathrm{ph}}$), and isotropy branch selection---provides a unified, mode-resolved framework for organizing and identifying emergent magnetic symmetries directly from parent-phase lattice vibrations.

To demonstrate the practical execution of this framework and illustrate how parent-group representations systematically construct candidate magnetic structures, we apply our methodology to monolayer $\mathrm{Cd}_2\mathrm{N}_3$. As detailed below, this symmetry mapping not only identifies the symmetry channel and real-space configuration of the lowest-energy magnetic ground state but also uncovers hidden multipolar and non-collinear magnetic manifolds directly from parent lattice kinematics.

% --- Section V: Case Study ---
\section{Symmetry Analysis and Magnetic Templates of \texorpdfstring{Cd$_2$N$_3$}{Cd2N3} Monolayer}
\label{sec:case_study}

The interplay between lattice dynamics and electronic spin degrees of freedom is fundamentally governed by crystalline symmetry. For monolayer Cd$_2$N$_3$ \cite{Zhang2021-Cd2N3}, the real-space atomic configuration and reciprocal lattice are depicted in Fig.~\ref{UCBZ}. Structural stability is confirmed by the phonon dispersion in Fig.~\ref{Cd2N3-phonon-band}, showing no imaginary frequencies across the entire Brillouin zone.

\begin{figure}[htbp]
\centering
\includegraphics[scale=0.40]{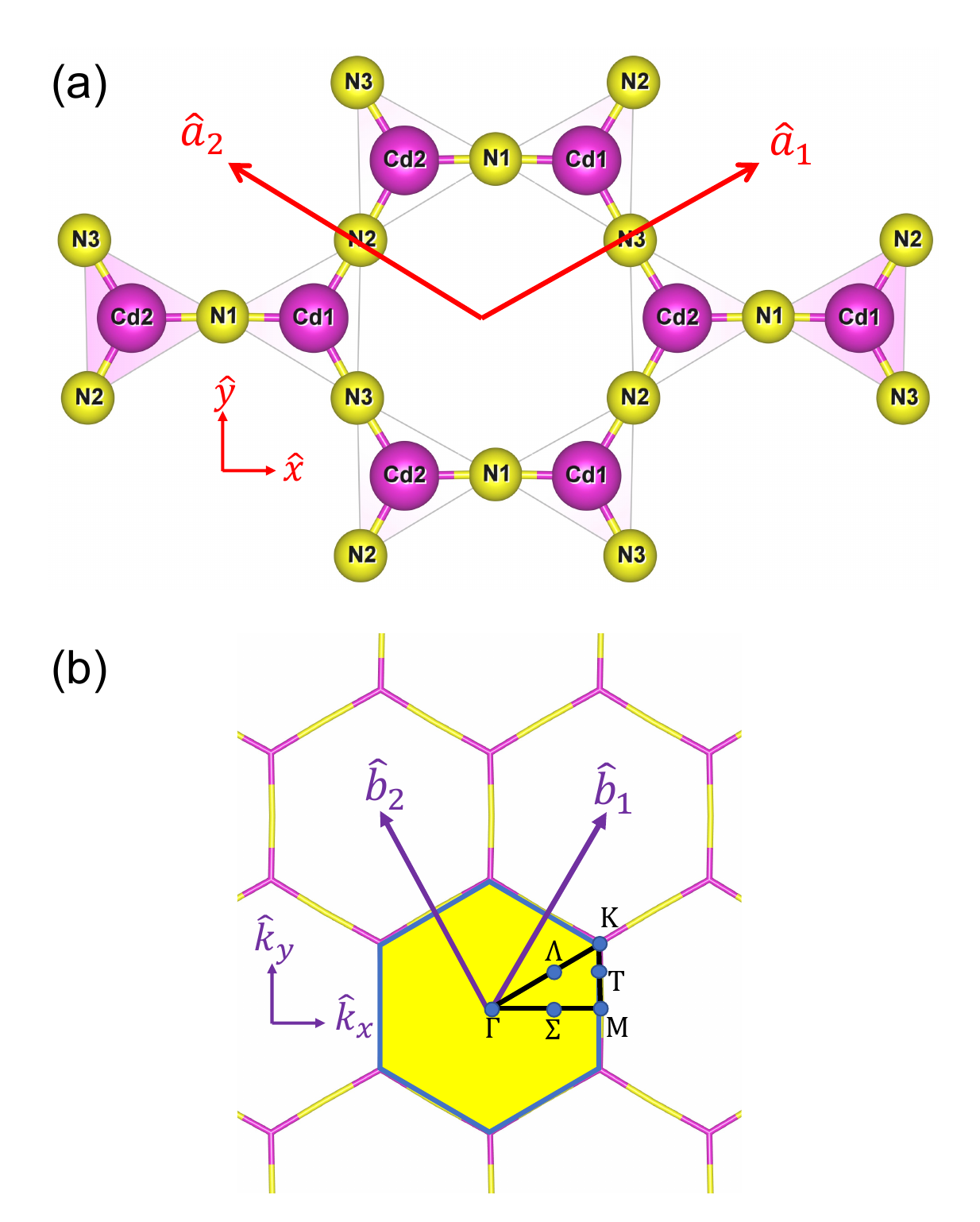}
\caption{\label{UCBZ} Real (a) and reciprocal (b) lattices for two-dimensional Cd$_2$N$_3$. The lattice vectors for real space and the two nonequivalent Cd atoms and three N atoms are indicated in (a); the lattice vectors for reciprocal space and high-symmetry points in the first Brillouin zone are indicated in (b).}
\end{figure}

\begin{figure}[htbp]
\centering
\includegraphics[scale=0.54]{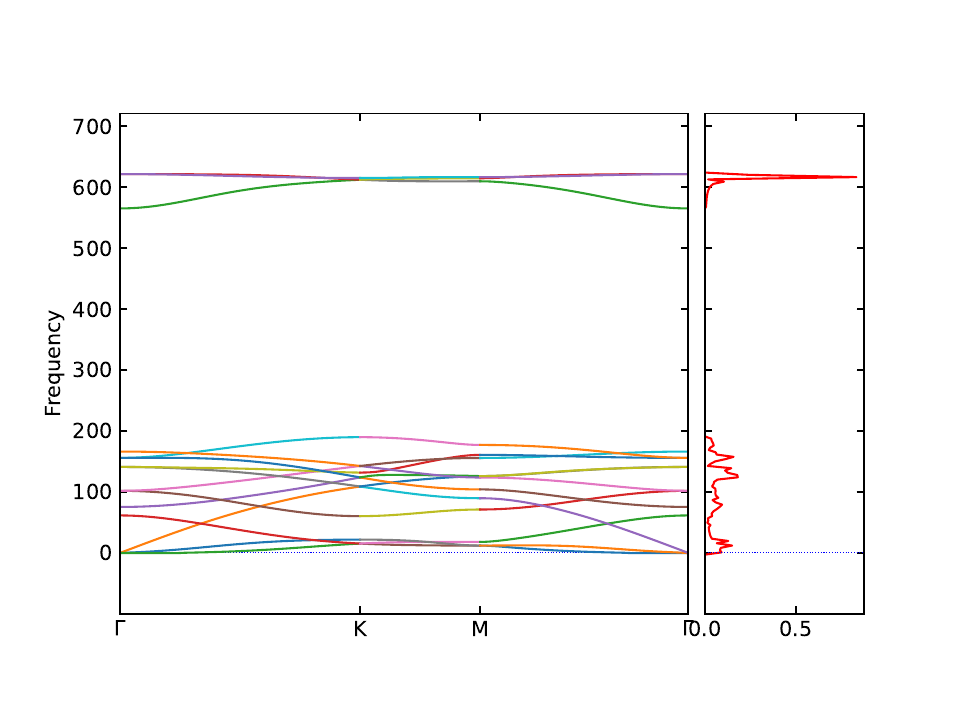}
\caption{\label{Cd2N3-phonon-band} Dispersion relation and density of states of monolayer Cd$_2$N$_3$ in units of cm$^{-1}$.}
\end{figure}

\begin{table}[htbp]
\centering
\begin{ruledtabular}
\caption{\label{Reduced Representation-D6h} Characters of vector, equivalent, and vibration representations for Cd$_2$N$_3$ as well as characters for Cd and N atoms on Wyckoff sites $2c$ and $3f$.}
\begin{tabular}{ccccccccccccc}
 $D_{6h}$ &$E$ &$ 2C_{6}$&$ 2C_{3}$&$C_{2}$&$3C'_{2}$ &$3C''_{2}$&$i$ &$2S_{3}$ &$2S_{6}$ &$\sigma_{h}$ &$3\sigma_{v}$&$3\sigma_{d}$\\
 \hline
 $\chi_{\rm{2Cd}}(2c)$ & 2 & 0 & 2 & 0 & 2 & 0 & 0 & 2 & 0 & 2 & 2 & 0 \\
 $\chi_{\rm{3N}}(3f)$ & 3 & 0 & 0 & 3 & 1 & 1 & 3 & 0 & 0 & 3 & 1 & 1 \\
 $\chi_{\rm{equ}}$ & 5 & 0 & 2 & 3 & 3 & 1 & 3 & 2 & 0 & 5 & 3 & 1 \\
 $\chi_{\rm{vec}}$ & 3 & 2 & 0 & -1 & -1 & -1 & -3 & -2 & 0 & 1 & 1 & 1 \\
 $\chi_{\rm{lat.vib.}}$ & 15 & 0 & 0 & -3 & -3 & -1 & -9 & -4 & 0 & 5 & 3 & 1 \\
\end{tabular}
\end{ruledtabular}
\end{table}

\subsection{Phonon Decomposition, Sublattice Resolution, and Pseudoscalar Mapping}
\label{subsec:phonon_decomp}

The structural symmetry of monolayer Cd$_2$N$_3$ belongs to space group $P6/mmm$ (No. 191) with $D_{6h}$ point group symmetry at $\Gamma$. Two Cd atoms occupy $2c$ Wyckoff sites, and three N atoms reside at $3f$ sites. To resolve the atomic origins and spatial directional components of the vibrational modes, we decompose the lattice vibration representation $\Gamma_{\rm{lat.vib.}}$ into irreps of $D_{6h}$ using the character atlas in Table~\ref{Reduced Representation-D6h}:
\begin{align}
\Gamma_{\rm{polar}} &= A_{2u}(z) \oplus E_{1u}(x,y), \\
\Gamma_{\rm{eq}} &= \Gamma_{\rm{2Cd}} \oplus \Gamma_{\rm{3N}} = (A_{1g} \oplus B_{1u}) \oplus (A_{1g} \oplus E_{2g}).
\end{align}
The 15-dimensional $\Gamma$-point vibrational representation corresponds to the $3N=15$ Cartesian vibrational degrees of freedom of the five-atom unit cell, equivalently represented by 15 phonon eigenvectors of the $\Gamma$-point dynamical matrix:
\begin{align}
\Gamma_{\mathrm{ph}}^{\mathrm{tot}} &= 2A_{2u} \oplus B_{1u} \oplus B_{2u} \oplus B_{2g} \oplus 3E_{1u} \oplus E_{2u} \oplus E_{2g},
\end{align}
partitioned into acoustic ($\Gamma_{\rm{aco}} = A_{2u} \oplus E_{1u}$) and optical ($\Gamma_{\rm{opt}} = A_{2u} \oplus 2E_{1u} \oplus B_{1u} \oplus B_{2u} \oplus B_{2g} \oplus E_{2u} \oplus E_{2g}$) branches. 

To establish the site- and direction-specific atomic drivers for each mode, we project the displacement fields onto individual sublattices and directional components:
\begin{align}
\Gamma_{\rm{2Cd}}^{z} &= A_{2u} \oplus B_{2g}, \quad \Gamma_{\rm{2Cd}}^{(x,y)} = E_{1u} \oplus E_{2g}, \\
\Gamma_{\rm{3N}}^{z} &= A_{2u} \oplus E_{2u}, \quad \Gamma_{\rm{3N}}^{(x,y)} = 2E_{1u} \oplus B_{1u} \oplus B_{2u}.
\end{align}
These sublattice projections provide the explicit kinematic foundation for explaining why specific magnetic modes are localized on distinct atomic species and spatial orientations.

For monolayer $\mathrm{Cd}_2\mathrm{N}_3$, first-principles magnetic calculations show that both Cd and N sites carry nonzero magnetic moments, although the Cd moments are comparatively weak. Thus, the magnetic sublattice coincides with the full crystallographic basis ($\mathcal{S}_{\mathrm{mag}} = \mathcal{S}_{\mathrm{full}}$), and the site-restriction map $Q_{\mathrm{mag}}$ defined in Sec.~\ref{subsec:sublattice} reduces to the identity operator on the full site-coordinate space ($Q_{\mathrm{mag}} = I$). Consequently, the full-cell representation $\Gamma_{\mathrm{mag}}^{\mathrm{tot}}$ coincides with the sublattice magnetic representation $\Gamma_{\mathrm{mag}}^{\mathrm{(sub)}}$ [Eq.~\eqref{eq:sub_mag}], and the phonon-derived templates retain the complete site-resolved information of the parent vibrational eigenvectors without multiplicity collapse.

Within point group $D_{6h}$, the pseudoscalar representation is uniquely identified as $\Gamma_{\mathrm{ps}} = A_{1u}$ (Table~\ref{tab:table1}). Applying the single-irrep mapping $\Gamma_{\text{mag}} = \Gamma \otimes A_{1u}$ [Eq.~\eqref{eq:irrep_map}] and the total representation mapping theorem $\Gamma_{\mathrm{mag}}^{\mathrm{tot}} = \Gamma_{\mathrm{ph}}^{\mathrm{tot}} \otimes A_{1u}$ [Eq.~\eqref{eq:tot_mag_map}], the 15-dimensional vibrational space maps directly onto the 15-dimensional magnetic configuration space ($\dim\Gamma_{\mathrm{mag}}^{\mathrm{tot}}=\dim\Gamma_{\mathrm{ph}}^{\mathrm{tot}}=15$):
\begin{align}
\Gamma_{\mathrm{mag}}^{\mathrm{tot}} &= 2A_{2g} \oplus B_{1g} \oplus B_{2g} \oplus B_{2u} \oplus 3E_{1g} \oplus E_{2g} \oplus E_{2u}.
\end{align}
In strict accordance with the multiplicity preservation theorem [Eq.~\eqref{eq:template1}], the dimension of each symmetry sector is preserved.

To analyze and categorize these mapped states, Table~\ref{D6h_mapping_dictionary} presents the universal augmented symmetry-mapping dictionary for the $D_{6h}$ point group, systematically compiling for every irrep its pseudoscalar partner $\Gamma_{\mathrm{mag}}$, the resulting Shubnikov magnetic point group (MPG), and the path-dependent splitting into maximal isotropy subgroups (epikernels).

Applying this universal mapping atlas to monolayer $\mathrm{Cd}_2\mathrm{N}_3$ generates the concrete, mode-resolved magnetic classification summarized in Table~\ref{Table-Phon-Mag}. For each mode, Table~\ref{Table-Phon-Mag} explicitly details the structural driver derived from sublattice projections, order-parameter direction, symmetry-adapted basis vector, Shubnikov magnetic point group, and the resulting magnetic order type. The corresponding real-space spin vector configurations and residual magnetic symmetries are depicted in Fig.~\ref{fig:stability_multipoles} (for 1D and acoustic modes) and Fig.~\ref{fig:2D_manifolds_vertical} (for 2D optical manifolds), forming a complete real-space template atlas.

% ------------------- Table III  ----------------
\begin{table*}[t]
\centering
\renewcommand{\arraystretch}{1.25}
\caption{Augmented symmetry-mapping dictionary for the point group $D_{6h}$ ($6/mmm$). The pseudoscalar basis is defined as $\mathcal{P}\equiv\det(\mathbf{R})$.
For two-dimensional irreducible representations (e.g., $E_{1g}$ and $E_{2u}$), the mapping generates split maximal isotropy magnetic subgroups (epikernels) corresponding to high-symmetry order-parameter directions (Paths B1 and B2).
Magnetic point groups (MPGs) are listed in both Schoenflies and Hermann--Mauguin (H--M) notations. To preserve the correspondence between symmetry elements of the parent group $D_{6h}=6/m_{h}m_{v}m_{d}$ and its isotropy subgroups, the H--M symbols of the resulting $D_{2h}$ magnetic subgroups follow the same mirror-plane ordering convention $(m_{h},m_{v},m_{d})$ as adopted for the parent group. Type-I magnetic point groups are marked with an asterisk (\text{*}).}
\label{D6h_mapping_dictionary}
\begin{tabular}{
l cccccccccccc 
l l 
c 
>{\centering\arraybackslash}p{2.5cm} 
>{\centering\arraybackslash}p{2.5cm}
}
\toprule
$D_{6h}$ & $E$ & $2C_6$ & $2C_3$ & $C_2$ & $3C_2'$ & $3C_2''$ & $i$ & $2S_3$ & $2S_6$ & $\sigma_h$ & $3\sigma_v$ & $3\sigma_d$ 
& \textbf{Basis} & $\bm{\Gamma}_{\mathrm{mag}}$ & \textbf{Path} 
& \textbf{MPG (S)} & \textbf{MPG (H-M)} \\
\midrule
$A_{1g}$ & 1 & 1 & 1 & 1 & 1 & 1 & 1 & 1 & 1 & 1 & 1 & 1 & $x^2+y^2, z^2$ & $A_{1u}$ & (1) & $D_{6h}(D_6)$ & $6/m'm'm'$ \\
$A_{2g}$ & 1 & 1 & 1 & 1 & -1 & -1 & 1 & 1 & 1 & 1 & -1 & -1 & $R_z$ & $A_{2u}$ & (1) & $D_{6h}(C_{6v})$ & $6/m'mm$ \\
$B_{1g}$ & 1 & -1 & 1 & -1 & 1 & -1 & 1 & -1 & 1 & -1 & 1 & -1 & - & $B_{1u}$ & (1) & $D_{6h}(D_{3h})$ & $6'/mmm'$ \\
$B_{2g}$ & 1 & -1 & 1 & -1 & -1 & 1 & 1 & -1 & 1 & -1 & -1 & 1 & - & $B_{2u}$ & (1) & $D_{6h}(D_{3h})$ & $6'/mm'm$ \\
$E_{1g}$ & 2 & 1 & -1 & -2 & 0 & 0 & 2 & 1 & -1 & -2 & 0 & 0 & $(R_x, R_y)$ & $E_{1u}$ & \makecell{Path B1 \\ Path B2} & \makecell{$D_{2h}(C_{2v})$ \\ $D_{2h}(C_{2v})$} & \makecell{$mm'm$ \\ $mmm'$} \\
\rowcolor{gray!10}
$E_{2g}$ & 2 & -1 & -1 & 2 & 0 & 0 & 2 & -1 & -1 & 2 & 0 & 0 & $(x^2-y^2, xy)$ & $E_{2u}$ & \makecell{Path B1 \\ Path B2} & \makecell{$D_{2h}(C_{2v})$ \\ $D_{2h}(D_{2})$} & \makecell{$m'mm$ \\ $m'm'm'$} \\
\midrule
\rowcolor{blue!10}
$\mathbf{A_{1u}}$ & \textbf{1} & \textbf{1} & \textbf{1} & \textbf{1} & \textbf{1} & \textbf{1} & \textbf{-1} & \textbf{-1} & \textbf{-1} & \textbf{-1} & \textbf{-1} & \textbf{-1} & $\mathcal{P}$ & $A_{1g}$ & (1) & $D_{6h}$ & $6/mmm^*$ \\
$A_{2u}$ & 1 & 1 & 1 & 1 & -1 & -1 & -1 & -1 & -1 & -1 & 1 & 1 & $z$ & $A_{2g}$ & (1) & $D_{6h}(C_{6h})$ & $6/mm'm'$ \\
$B_{1u}$ & 1 & -1 & 1 & -1 & 1 & -1 & -1 & 1 & -1 & 1 & -1 & 1 & - & $B_{1g}$ & (1) & $D_{6h}(D_{3d})$ & $6'/m'mm'$ \\
$B_{2u}$ & 1 & -1 & 1 & -1 & -1 & 1 & -1 & 1 & -1 & 1 & 1 & -1 & - & $B_{2g}$ & (1) & $D_{6h}(D_{3d})$ & $6'/m'm'm$ \\
$E_{1u}$ & 2 & 1 & -1 & -2 & 0 & 0 & -2 & -1 & 1 & 2 & 0 & 0 & $(x, y)$ & $E_{1g}$ & \makecell{Path B1 \\ Path B2} & \makecell{$D_{2h}(C_{2h})$ \\ $D_{2h}(C_{2h})$} & \makecell{$m'mm'$ \\ $m'm'm$} \\
$E_{2u}$ & 2 & -1 & -1 & 2 & 0 & 0 & -2 & 1 & 1 & -2 & 0 & 0 & - & $E_{2g}$ & \makecell{Path B1 \\ Path B2} & \makecell{$D_{2h}(C_{2h})$ \\ $D_{2h}$} & \makecell{$mm'm'$ \\ $mmm^*$} \\
\bottomrule
\end{tabular}
\end{table*}

% ----------- Table IV --------------------

\begin{table*}[t]
\centering
\caption{\label{Table-Phon-Mag} Mode-resolved classification of the 15-dimensional $\Gamma$-point vibrational representation and its corresponding symmetry-adapted magnetic sectors for $\mathrm{Cd}_2\mathrm{N}_3$. Abbreviations: Aco and Opt denote acoustic and optical branches. Path denotes the chosen high-symmetry order-parameter direction (OPD) within the representation manifold. HM Symbol: Hermann-Mauguin notation for Shubnikov magnetic point group. Type-I groups are marked with an asterisk (\text{*}).}
\begin{tabular}{lcccccc}
\toprule
Category & $\Gamma_{\text{ph}}$ & Path & Symmetry-Adapted Magnetic Basis & Symmetry Path $G  \xrightarrow{\Gamma_{\text{mag}}}  G'(H)$ & HM Symbol & Mag. Order \\
\midrule
Aco & $A_{2u} $ & $(1)$ & Cd($z$), N($z$) & $D_{6h} \xrightarrow{A_{2g}} D_{6h}(C_{6h})$ & $6/mm'm'$ & FM-$z$ \\
\multirow{2}{*}{Aco} & \multirow{2}{*}{$E_{1u} $} & Path B1 & Cd, N ($x,y$) & $D_{6h} \xrightarrow{E_{1g}} D_{2h}(C_{2h})$ & $m'mm'$ & FM-$xy$ \\
 & & Path B2 & Cd, N ($x,y$) & $D_{6h} \xrightarrow{E_{1g}} D_{2h}(C_{2h})$ & $m'm'm$ & FM-$xy$ \\
\midrule
Opt & $A_{2u} $ & $(1)$ & Cd($z$), N($-z$) & $D_{6h} \xrightarrow{A_{2g}} D_{6h}(C_{6h})$ & $6/mm'm'$ & FiM-$z$ \\
\multirow{2}{*}{Opt} & \multirow{2}{*}{$2E_{1u} $} & Path B1 & Cd, N ($x,y$) & $D_{6h} \xrightarrow{E_{1g}} D_{2h}(C_{2h})$ & $m'mm'$ & FiM-$xy$ \\
 & & Path B2 & Cd, N ($x,y$) & $D_{6h} \xrightarrow{E_{1g}} D_{2h}(C_{2h})$ & $m'm'm$ & FiM-$xy$ \\
Opt & $B_{1u}$ & $(1)$ & N (Radial)  & $D_{6h} \xrightarrow{B_{1g}} D_{6h}(D_{3d})$ & $6'/m'mm'$ & NC-AFM \\
Opt & $B_{2u}$ & $(1)$ & N (Tangential) & $D_{6h} \xrightarrow{B_{2g}} D_{6h}(D_{3d})$ & $6'/m'm'm$ & NC-AFM \\
Opt & $B_{2g}$ & $(1)$ & Cd(Staggered $z$) & $D_{6h} \xrightarrow{B_{2u}} D_{6h}(D_{3h})$ & $6'/mm'm$ & AFM-$z$ \\
\multirow{2}{*}{Opt} & \multirow{2}{*}{$E_{2u} $} &  Path B1 & N (Staggered $z$) & $D_{6h} \xrightarrow{E_{2g}} D_{2h}(C_{2h})$ & $mm'm'$ & AFM-$z$ \\
 & &  Path B2 & N (Staggered $z$) & $D_{6h} \xrightarrow{E_{2g}} D_{2h}$ & $mmm$\textsuperscript{\textbf{*}} & AFM-$z$ \\
\multirow{2}{*}{Opt} & \multirow{2}{*}{$E_{2g} $} & Path B1 & Cd ($x,y$) & $D_{6h} \xrightarrow{E_{2u}} D_{2h}(C_{2v})$ & $m'mm$ & AFM-$y$ \\
 & & Path B2 & Cd ($x,y$) & $D_{6h} \xrightarrow{E_{2u}} D_{2h}(D_{2})$ & $m'm'm'$ & AFM-$x$ \\
 \bottomrule
\end{tabular}
\end{table*}

\subsection{Physical Interpretation of Acoustic and Optical Branches}
\label{subsec:acoustic_optical}

In the specific crystallographic realization of monolayer $\mathrm{Cd}_2\mathrm{N}_3$, Table~\ref{Table-Phon-Mag} reveals a clear correspondence between vibrational branch origins and the resulting magnetic configurations:
\begin{itemize}
    \item \textbf{Acoustic Modes ($\Gamma_{\text{aco}} = A_{2u} \oplus E_{1u}$):} Involve rigid, in-phase translations of all atoms across the unit cell, mapping exclusively to uniform Ferromagnetic (FM) alignments [out-of-plane FM-$z$ in Fig.~\ref{fig:stability_multipoles}(a) and in-plane FM-$xy$ in Fig.~\ref{fig:stability_multipoles}(c, d)].
    \item \textbf{Optical Modes ($\Gamma_{\text{opt}}$):} Involve relative sublattice displacements and generate nonuniform magnetic templates; the mapped configurations in this material realize Ferrimagnetic (FiM), Antiferromagnetic (AFM), and non-collinear cluster multipolar textures [Fig.~\ref{fig:stability_multipoles}(b, e, f) and Fig.~\ref{fig:2D_manifolds_vertical}].
\end{itemize}
This acoustic/optical correspondence represents a material-specific realization of the present monolayer system and is not a general theorem that optical modes must always generate antiferromagnetic or ferrimagnetic order.

\subsection{One-Dimensional Modes: Uniquely Locked Spin Configurations}
\label{subsec:1d_modes}

For a nonzero order parameter belonging to a non-degenerate one-dimensional representation occurring with unit multiplicity ($d_\Gamma = 1, n_\Gamma = 1$) on the magnetic sublattice, spatial symmetry completely and uniquely fixes the relative spin directions and sign patterns across sublattices up to an overall amplitude scale and global sign [Sec.~\ref{subsec:decoupling}]. When multiple copies or distinct Wyckoff orbits are present (such as the $2A_{2g}$ sector in $\mathrm{Cd}_2\mathrm{N}_3$), spatial symmetry uniquely locks the intrasublattice spin alignment, while the parent phonon normal modes provide the physically resolved orthogonal basis vectors resolving the intersublattice sign channels:

\begin{itemize}
    \item \textbf{Out-of-Plane $A_{2u} \to A_{2g}$ Channels:} The acoustic $A_{2u}$ mode maps to uniform FM-$z$ [Fig.~\ref{fig:stability_multipoles}(a)]. In contrast, the optical $A_{2u}$ mode involves anti-phase Cd--N out-of-plane motion, mapping uniquely to an uncompensated collinear ferrimagnetic state (FiM-$z$) with magnetic point group $D_{6h}(C_{6h})$ ($6/mm'm'$) [Fig.~\ref{fig:stability_multipoles}(b)]. The symmetry mapping identifies this symmetry-compatible configuration, and first-principles calculations subsequently confirm that this FiM-$z$ state is the realized magnetic ground state of monolayer Cd$_2$N$_3$.
    \item \textbf{Cluster Magnetic Octupoles ($B_{1u}$ and $B_{2u}$):} The optical $B_{1u}$ driver maps to $B_{1g}$ [Fig.~\ref{fig:stability_multipoles}(e)], generating a radial all-in/all-out spin configuration on N atoms with $D_{6h}(D_{3d})$ ($6'/m'mm'$) symmetry. The $B_{2u}$ driver maps to $B_{2g}$ [Fig.~\ref{fig:stability_multipoles}(f)], producing a tangential vortex-like spin pattern on N atoms with $D_{6h}(D_{3d})$ ($6'/m'm'm$) symmetry. Both constitute rank-3 even-parity cluster magnetic octupoles. The strict radial and tangential orientation constraints on N ($3f$) sites are direct physical manifestations of the site-stabilizer covariance condition [Eq.~\eqref{eq:nodal_cond} in Appendix~\ref{subsec:projection_nodal_app}].
    \item \textbf{Staggered Out-of-Plane AFM ($B_{2g}$):} The optical $B_{2g}$ driver maps to $B_{2u}$ (Table~\ref{Table-Phon-Mag}), yielding a staggered out-of-plane AFM-$z$ configuration on Cd sites with $D_{6h}(D_{3h})$ ($6'/mm'm$) symmetry.
\end{itemize}

\begin{figure}[t]
\centering
\begin{minipage}{0.48\columnwidth}
  \centering \includegraphics[width=\textwidth]{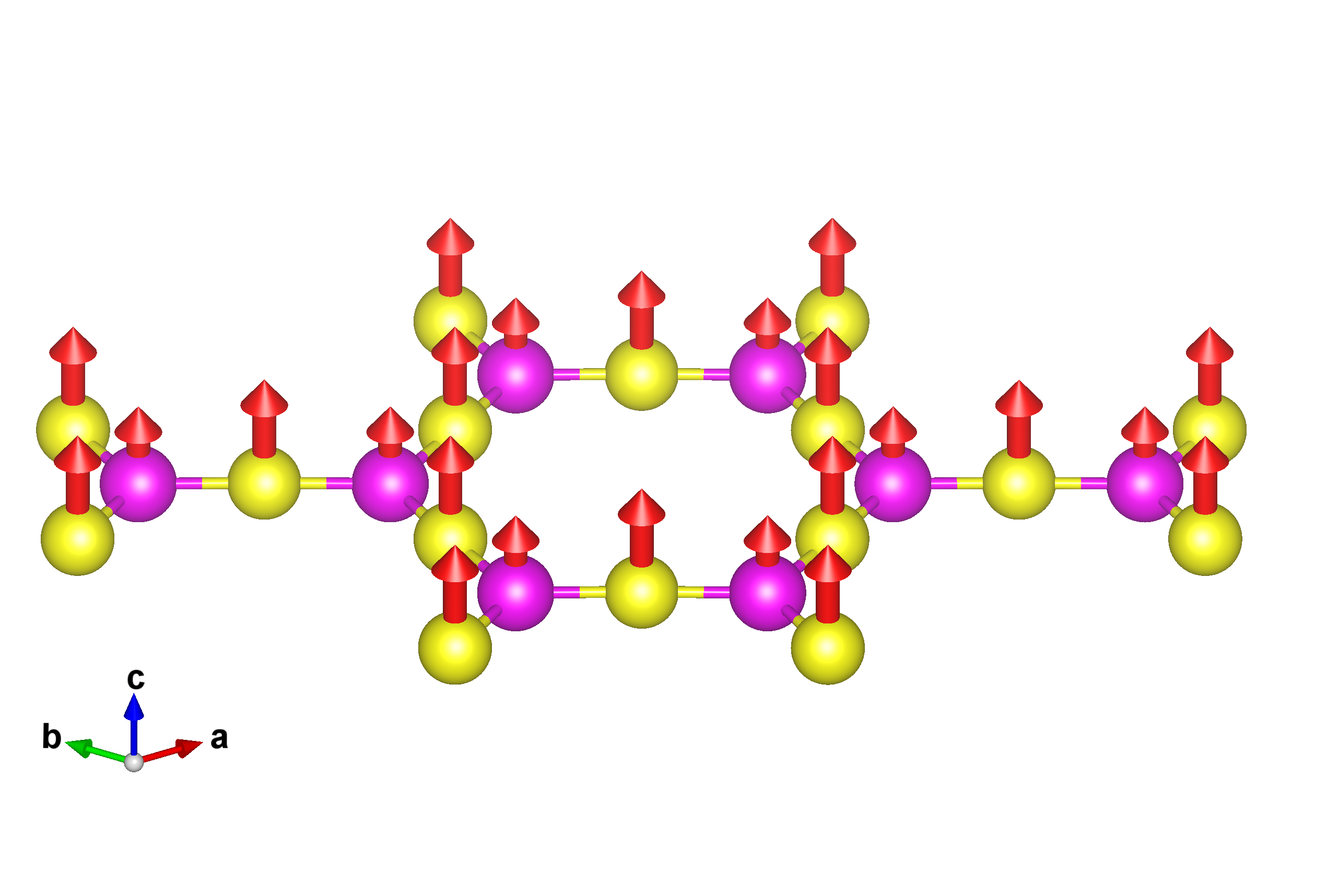} \\ 
  \scriptsize (a) $\Gamma_{\mathrm{ph}}: A_{2u}$ (Aco) \\ \tiny $\Gamma_{\mathrm{mag}}: A_{2g}$; $D_{6h}(C_{6h})$
\end{minipage}
\begin{minipage}{0.48\columnwidth}
  \centering \includegraphics[width=\textwidth]{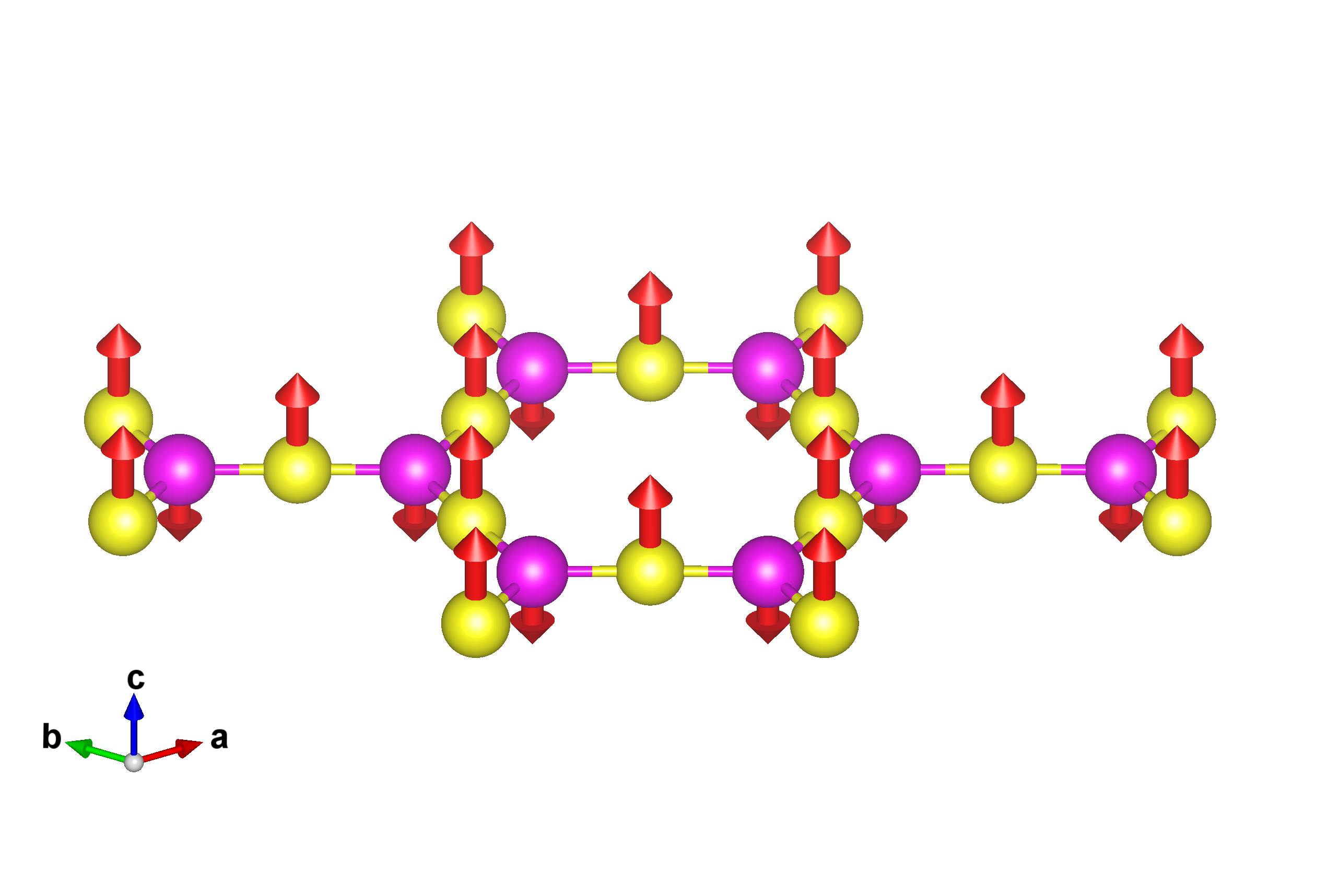} \\ 
  \scriptsize (b) $\Gamma_{\mathrm{ph}}: A_{2u}$ (Opt) \\ \tiny $\Gamma_{\mathrm{mag}}: A_{2g}$; $D_{6h}(C_{6h})$
\end{minipage} \\[0.3cm]
\begin{minipage}{0.48\columnwidth}
  \centering \includegraphics[width=\textwidth]{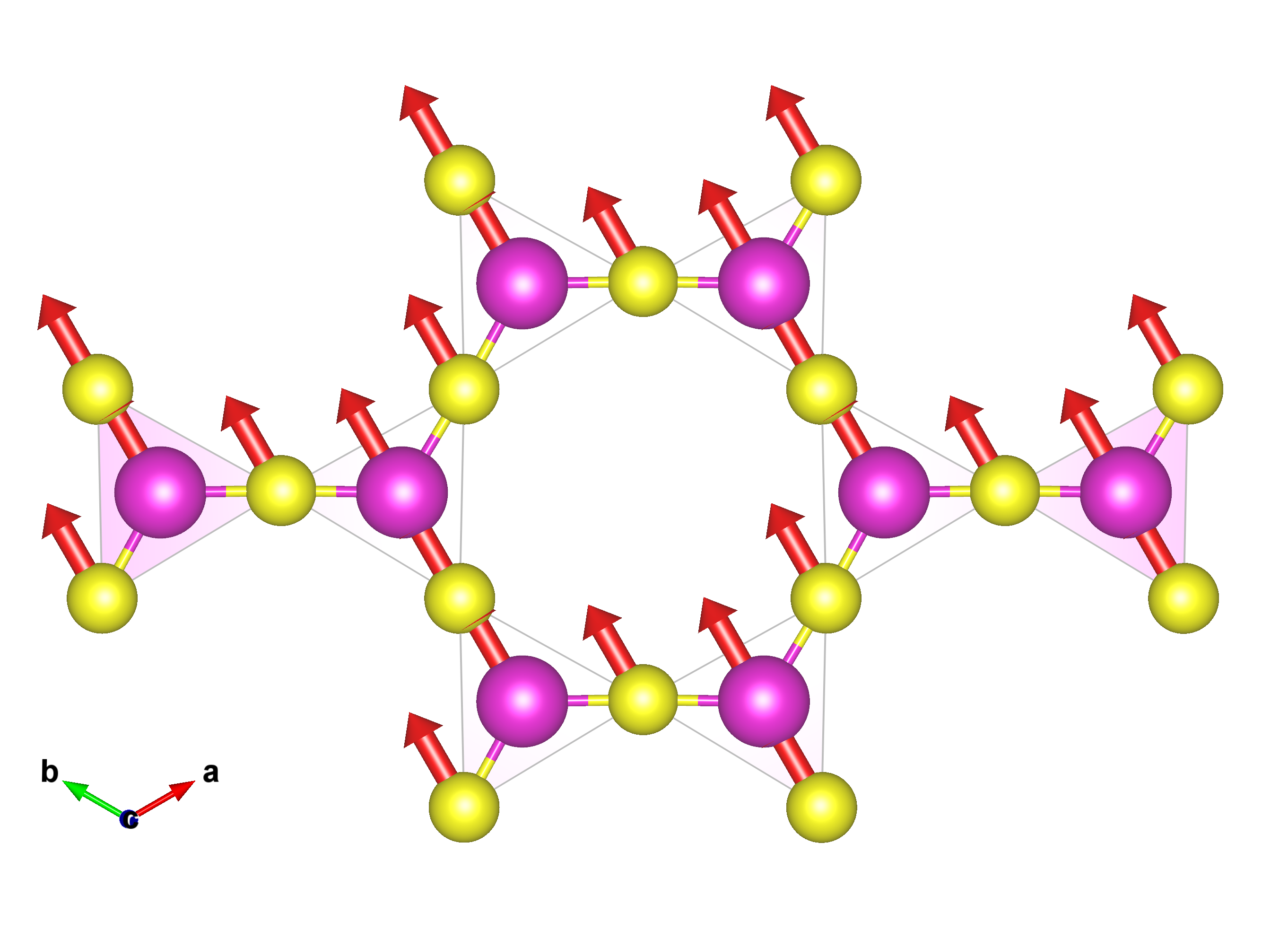} \\ 
  \scriptsize (c) $\Gamma_{\mathrm{ph}}: E_{1u}$ (Aco) B1 \\ \tiny $\Gamma_{\mathrm{mag}}: E_{1g}$; $D_{2h}(C_{2h})$
\end{minipage}
\begin{minipage}{0.48\columnwidth}
  \centering \includegraphics[width=\textwidth]{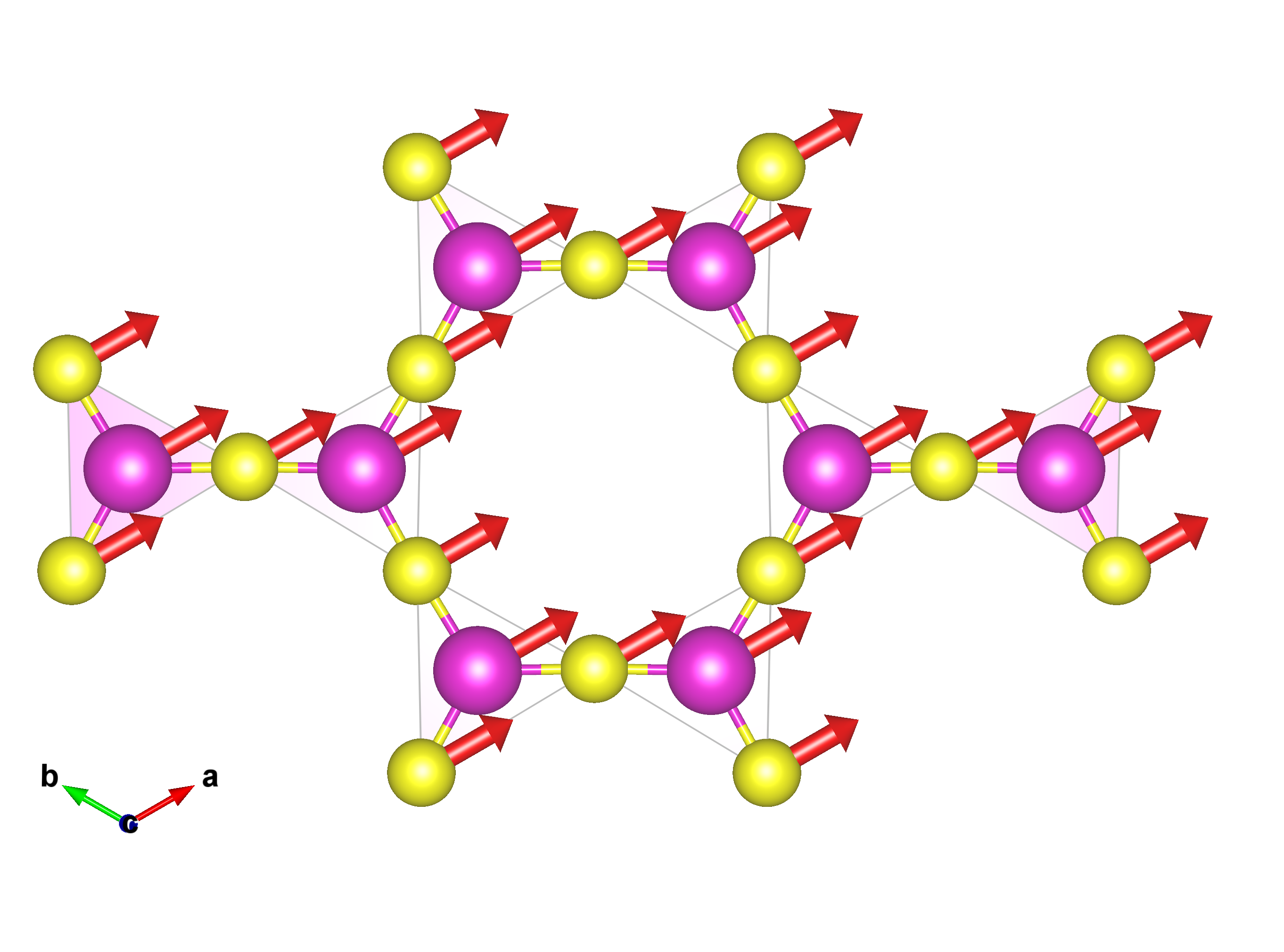} \\ 
  \scriptsize (d) $\Gamma_{\mathrm{ph}}: E_{1u}$ (Aco) B2 \\ \tiny $\Gamma_{\mathrm{mag}}: E_{1g}$; $D_{2h}(C_{2h})$
\end{minipage} \\[0.3cm]
\begin{minipage}{0.48\columnwidth}
  \centering \includegraphics[width=\textwidth]{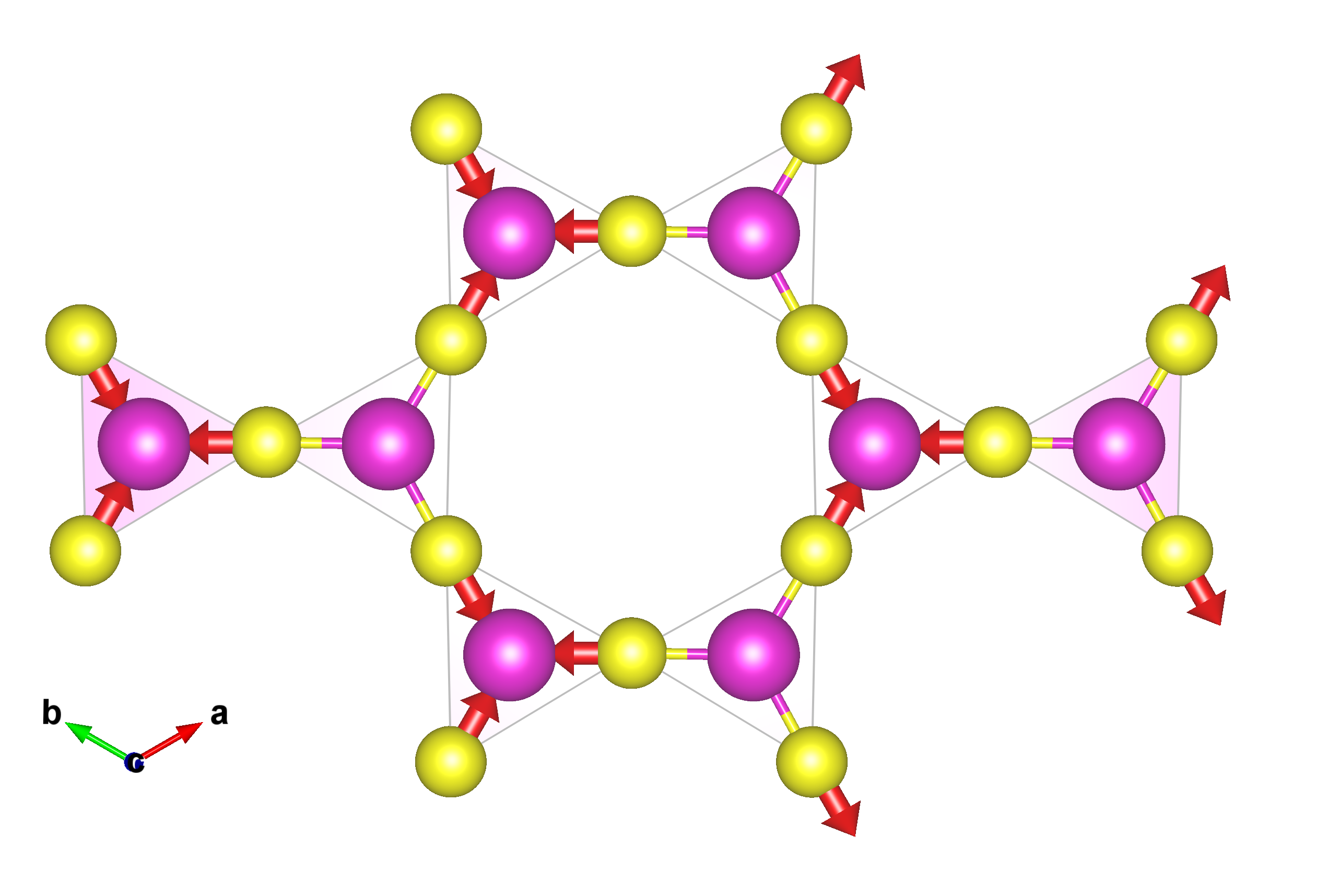} \\ 
  \scriptsize (e) $\Gamma_{\mathrm{ph}}: B_{1u}$ (Opt) \\ \tiny $\Gamma_{\mathrm{mag}}: B_{1g}$; $D_{6h}(D_{3d})$
\end{minipage}
\begin{minipage}{0.48\columnwidth}
  \centering \includegraphics[width=\textwidth]{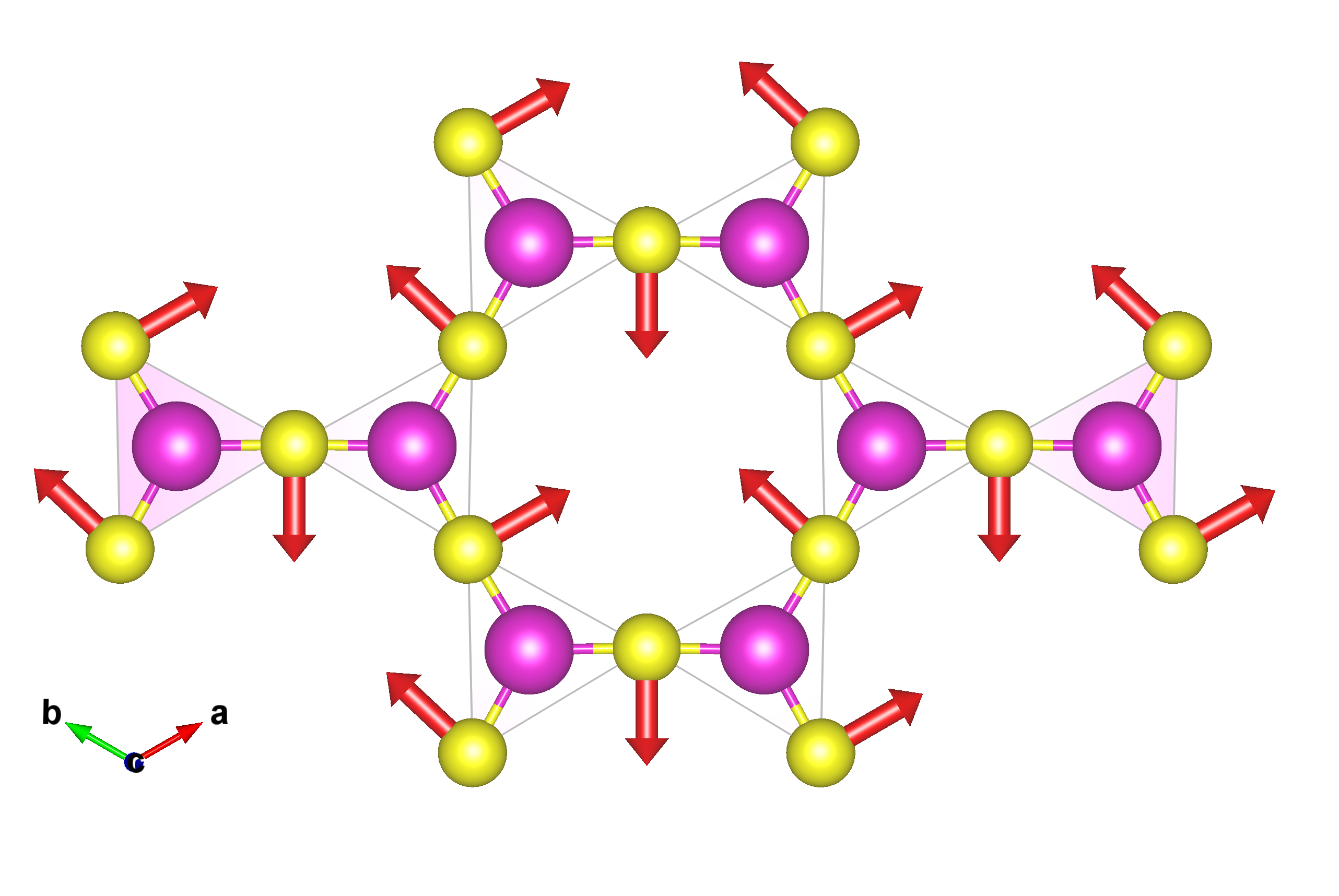} \\ 
  \scriptsize (f) $\Gamma_{\mathrm{ph}}: B_{2u}$ (Opt) \\ \tiny $\Gamma_{\mathrm{mag}}: B_{2g}$; $D_{6h}(D_{3d})$
\end{minipage}
\caption{Real-space symmetry-adapted magnetic templates generated from the three acoustic modes and three one-dimensional optical modes of monolayer $\mathrm{Cd}_2\mathrm{N}_3$ via the pseudoscalar mapping $\Gamma_{\mathrm{mag}} = \Gamma_{\mathrm{ph}} \otimes A_{1u}$. (a, c, d) The three acoustic vibrational degrees of freedom, comprising the out-of-plane acoustic mode $A_{2u}$ (a) and the two in-plane degenerate acoustic modes $E_{1u}$ (c, d), map exclusively to uniform out-of-plane FM-$z$ [$D_{6h}(C_{6h})$] and in-plane FM-$xy$ [$D_{2h}(C_{2h})$] ferromagnetic templates. (b, e, f) The three one-dimensional optical modes: the out-of-plane optical mode $A_{2u}$ (b) maps uniquely to the symmetry-compatible uncompensated FiM-$z$ template, which is subsequently confirmed to be the ground state by first-principles calculations, while the two in-plane non-degenerate optical modes $B_{1u}$ (e) and $B_{2u}$ (f) map to rank-3 even-parity cluster magnetic octupoles with radial all-in/all-out ($B_{1g}$) and toroidal vortex-like ($B_{2g}$) textures with residual $D_{6h}(D_{3d})$ symmetry.}
\label{fig:stability_multipoles}
\end{figure}

\begin{figure}[!t]
\centering
\begin{minipage}{0.48\columnwidth}
  \centering \includegraphics[width=\textwidth]{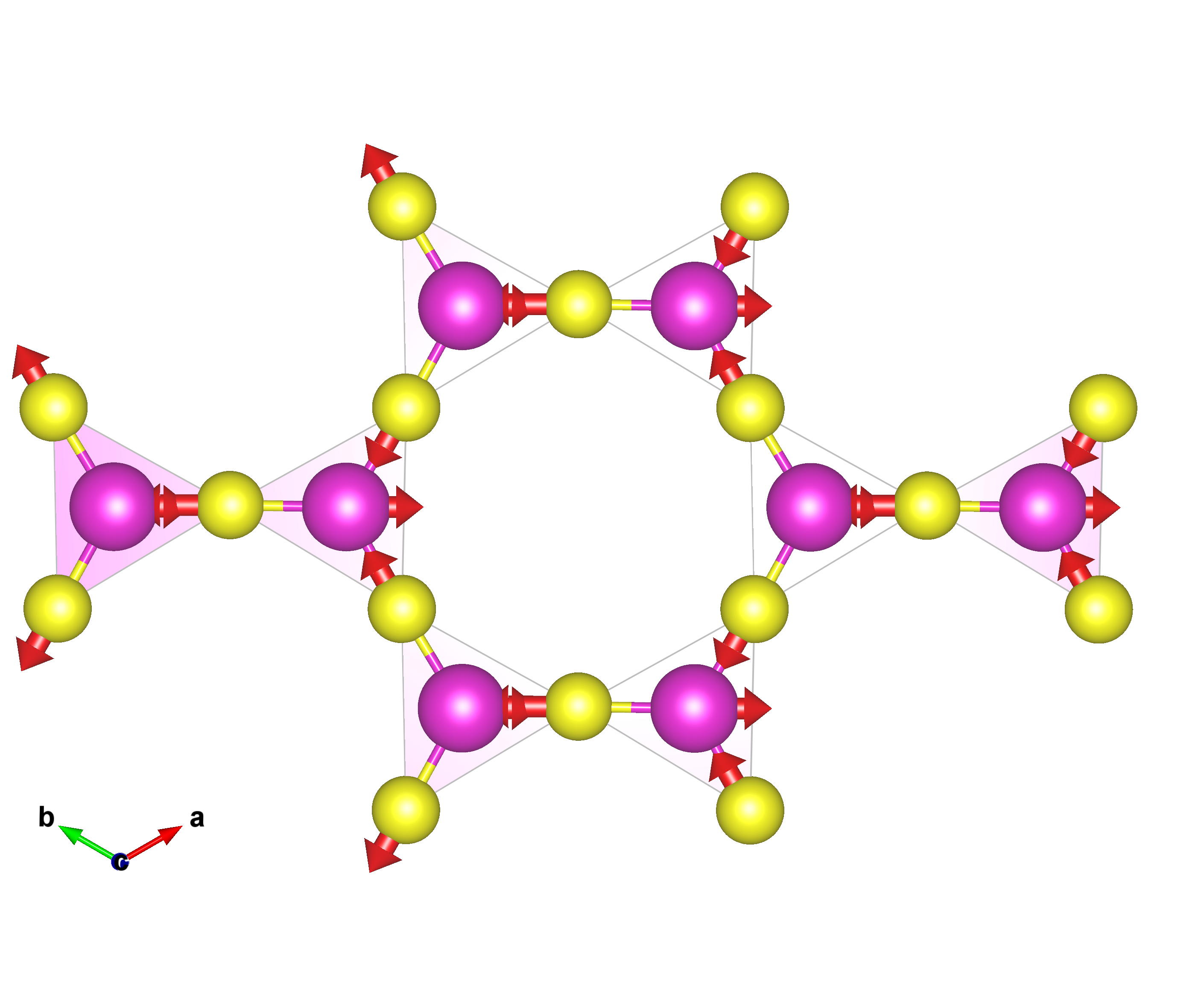} \\ 
  \scriptsize (a) $\Gamma_{\mathrm{ph}}: E_{1u}$ (Opt-1) B1 \\ \tiny $\Gamma_{\mathrm{mag}}: E_{1g}$; $D_{2h}(C_{2h})$
\end{minipage}
\begin{minipage}{0.48\columnwidth}
  \centering \includegraphics[width=\textwidth]{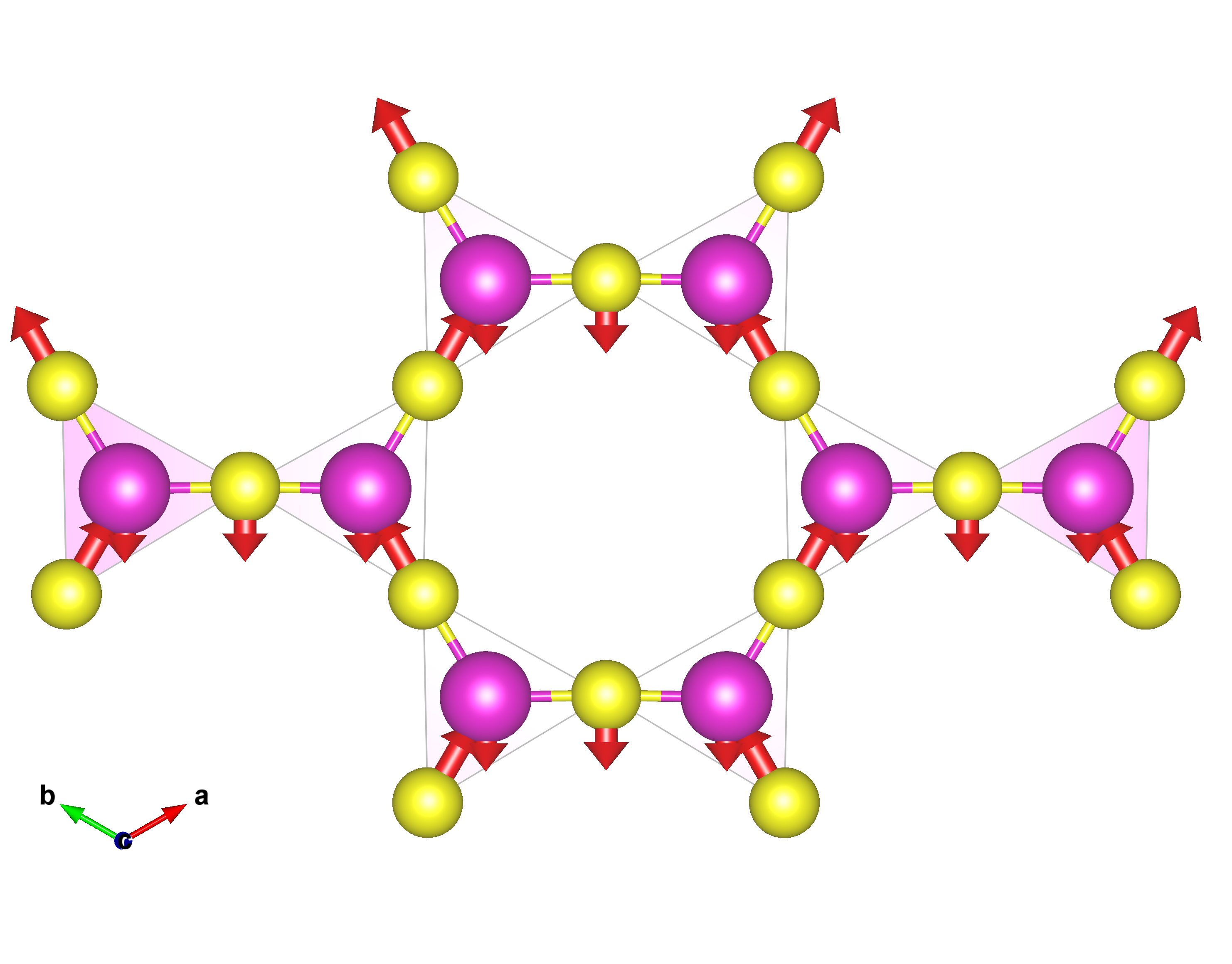} \\ 
  \scriptsize (b) $\Gamma_{\mathrm{ph}}: E_{1u}$ (Opt-1) B2 \\ \tiny $\Gamma_{\mathrm{mag}}: E_{1g}$; $D_{2h}(C_{2h})$
\end{minipage} \\[0.3cm]
\begin{minipage}{0.48\columnwidth}
  \centering \includegraphics[width=\textwidth]{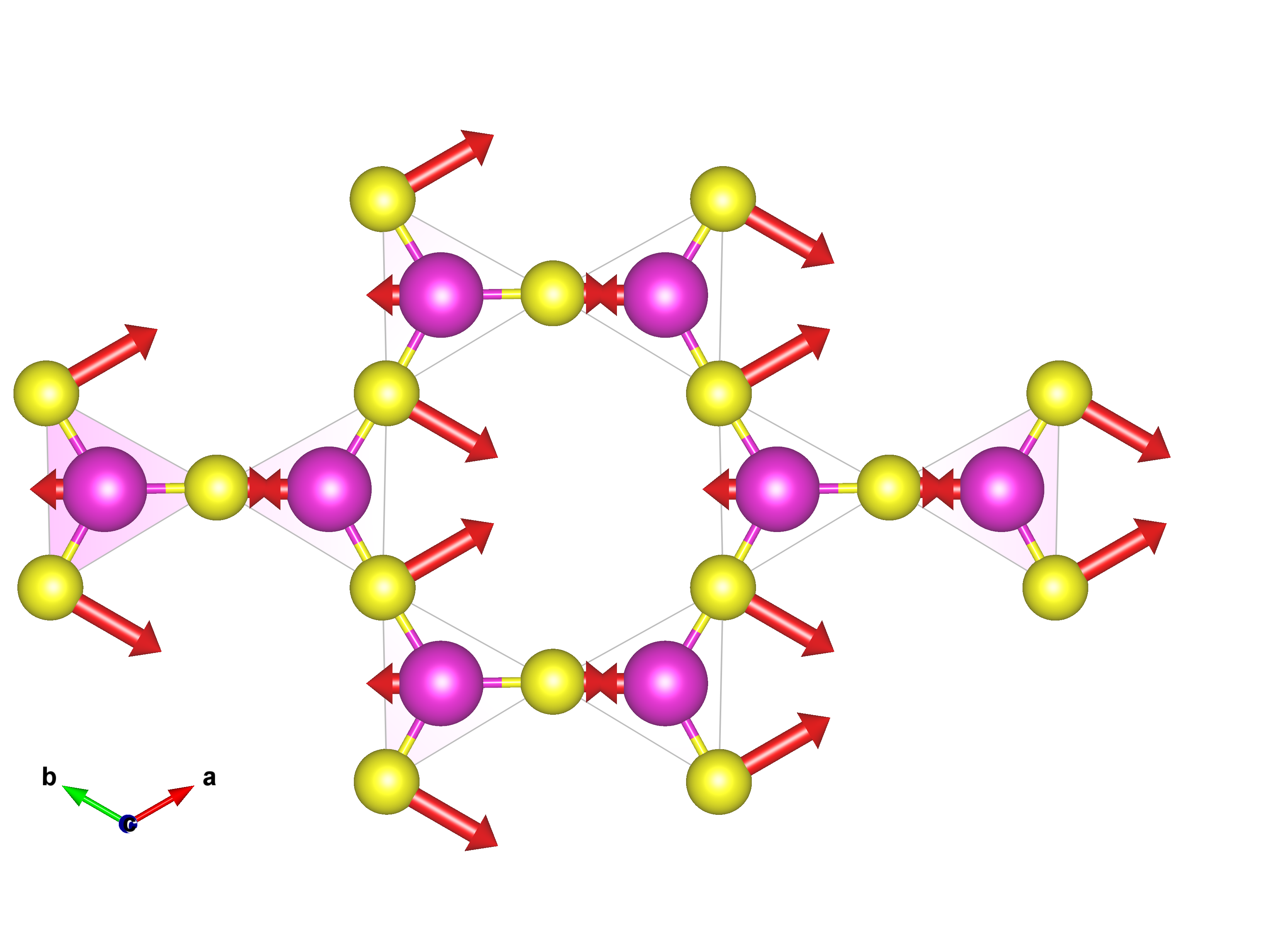} \\ 
  \scriptsize (c) $\Gamma_{\mathrm{ph}}: E_{1u}$ (Opt-2) B1 \\ \tiny $\Gamma_{\mathrm{mag}}: E_{1g}$; $D_{2h}(C_{2h})$
\end{minipage}
\begin{minipage}{0.48\columnwidth}
  \centering \includegraphics[width=\textwidth]{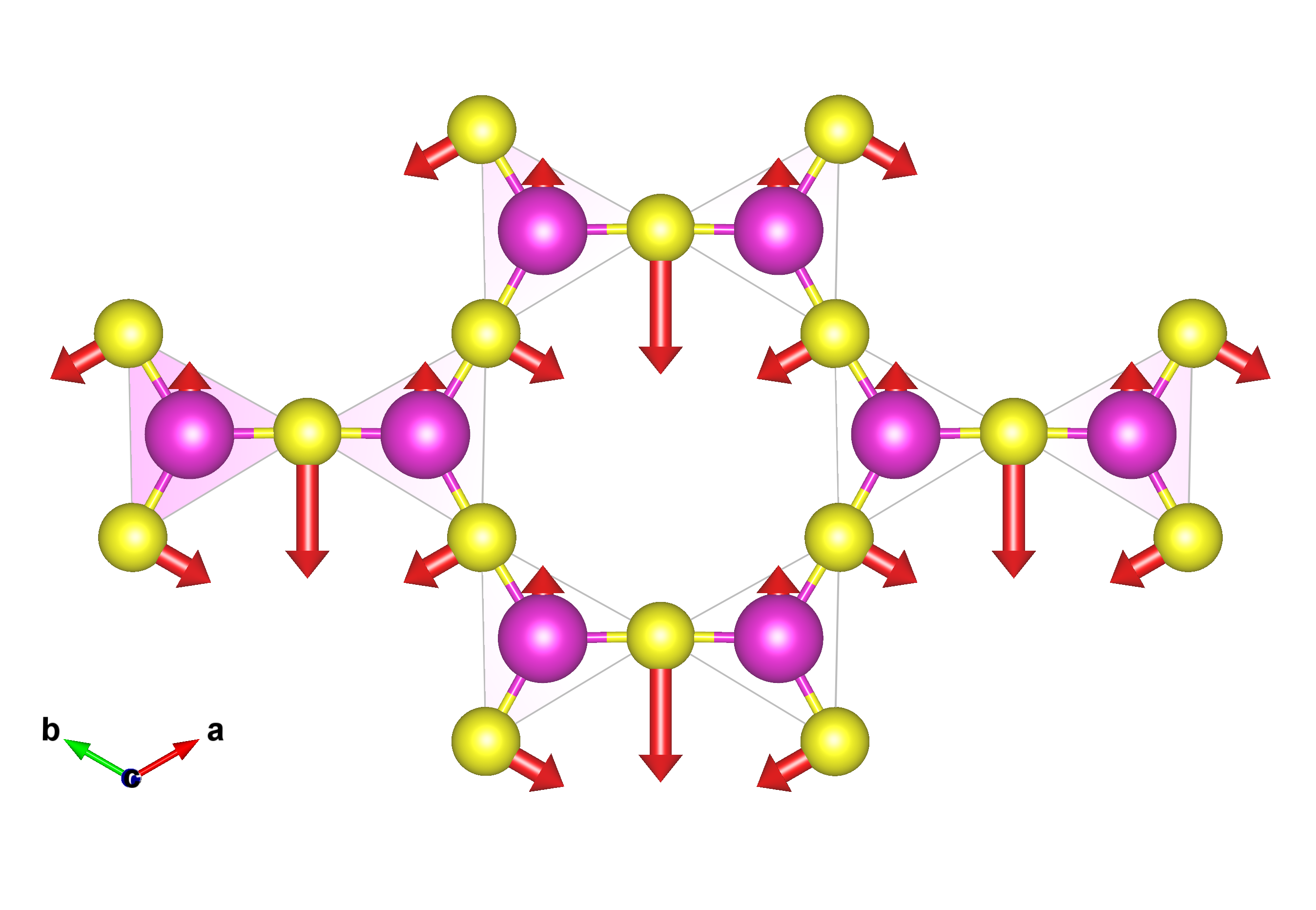} \\ 
  \scriptsize (d) $\Gamma_{\mathrm{ph}}: E_{1u}$ (Opt-2) B2 \\ \tiny $\Gamma_{\mathrm{mag}}: E_{1g}$; $D_{2h}(C_{2h})$
\end{minipage} \\[0.3cm]
\begin{minipage}{0.48\columnwidth}
  \centering \includegraphics[width=\textwidth]{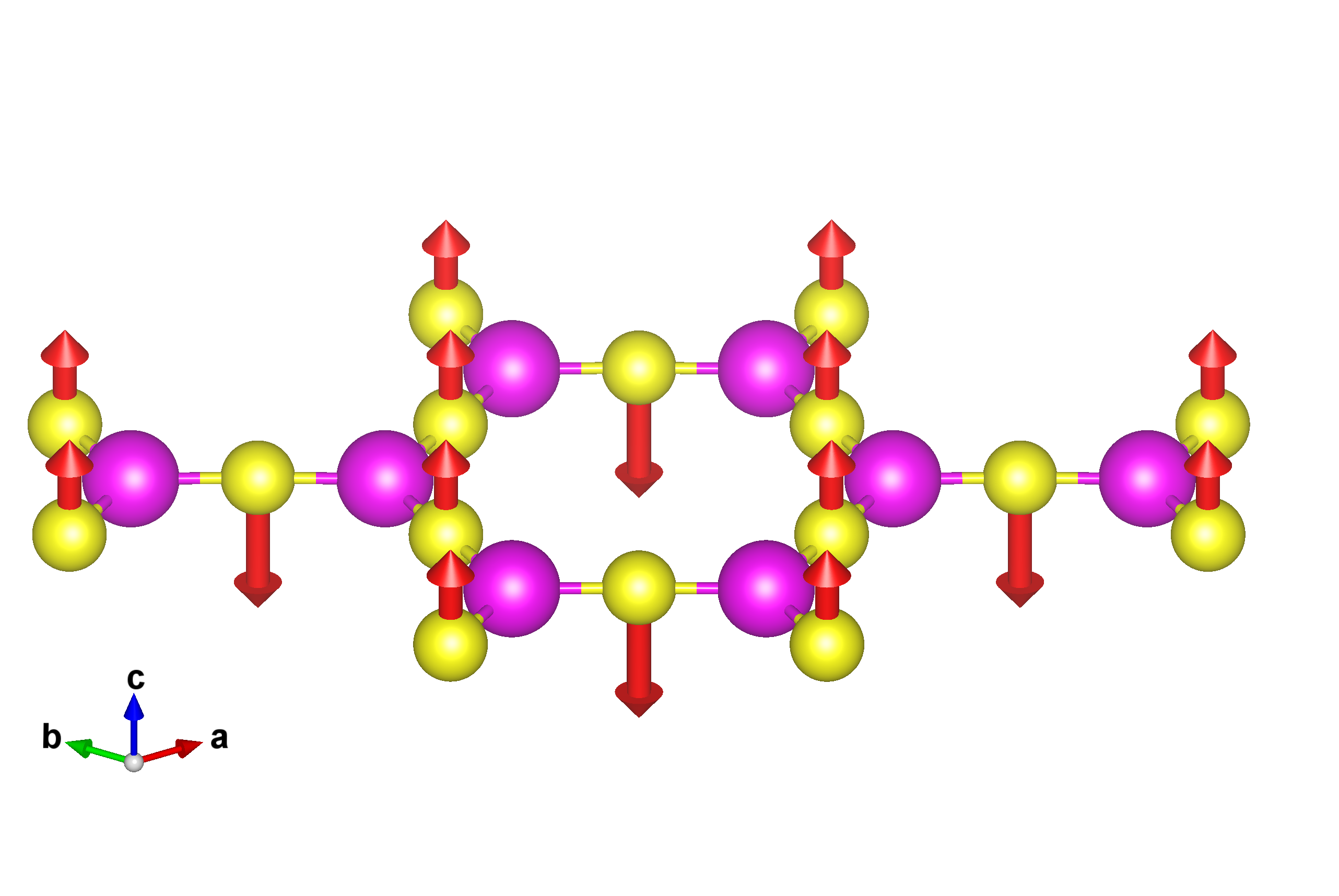} \\ 
  \scriptsize (e) $\Gamma_{\mathrm{ph}}: E_{2u}$ (Opt) B1 \\ \tiny $\Gamma_{\mathrm{mag}}: E_{2g}$; $D_{2h}(C_{2h})$
\end{minipage}
\begin{minipage}{0.48\columnwidth}
  \centering \includegraphics[width=\textwidth]{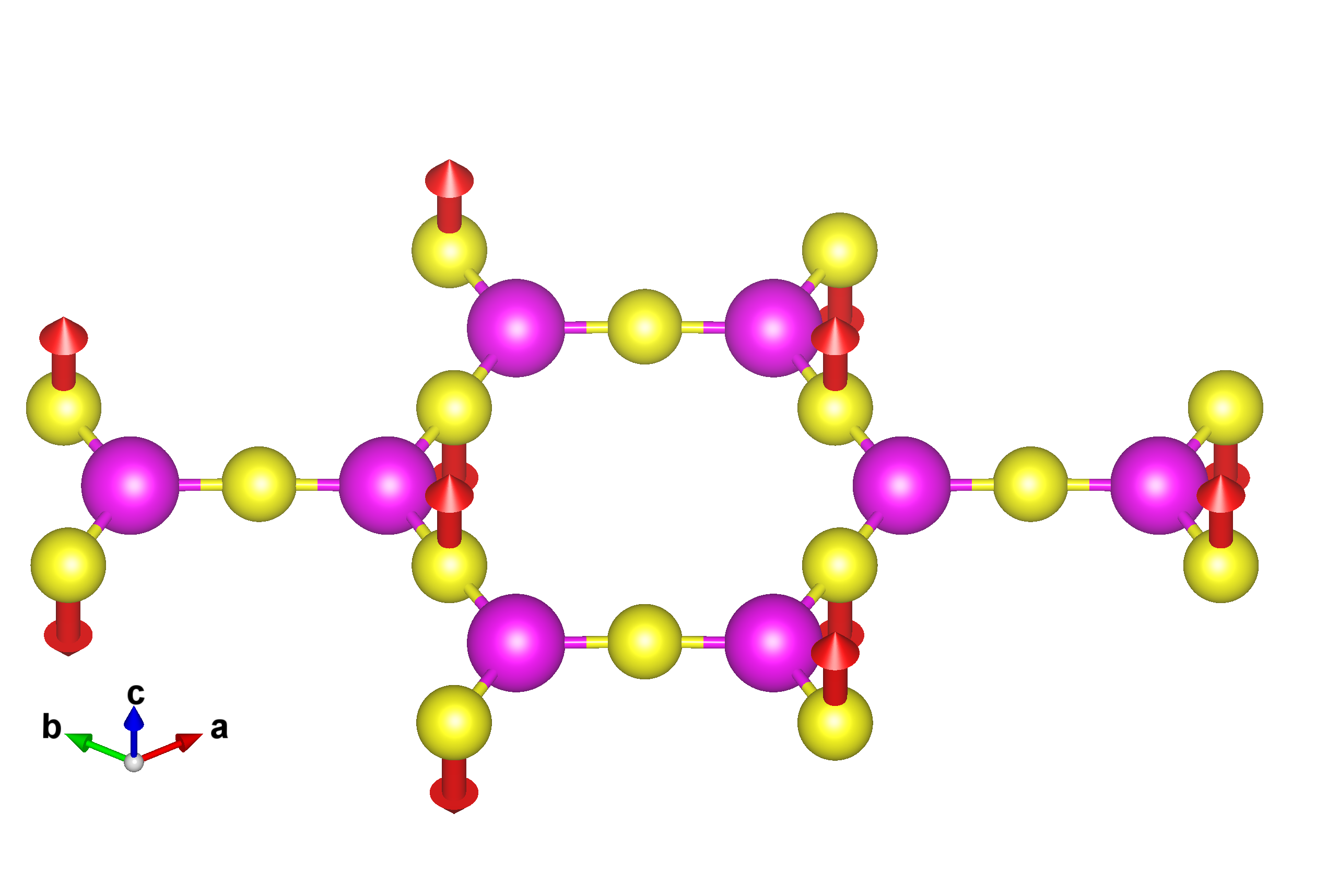} \\ 
  \scriptsize (f) $\Gamma_{\mathrm{ph}}: E_{2u}$ (Opt) B2\\ \tiny $\Gamma_{\mathrm{mag}}: E_{2g}$; $D_{2h}$
\end{minipage} \\[0.3cm]
\begin{minipage}{0.48\columnwidth}
  \centering \includegraphics[width=\textwidth]{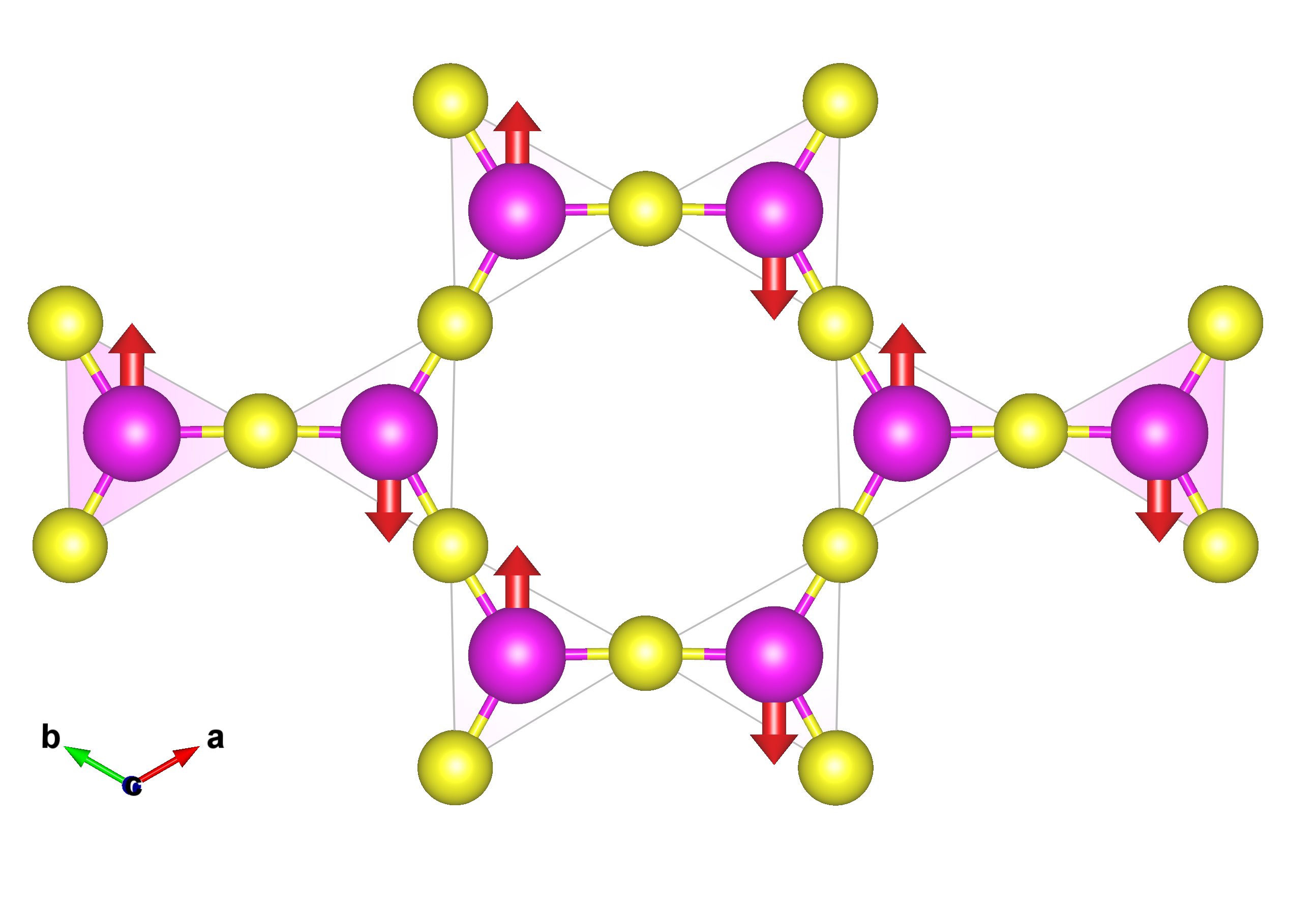} \\ 
  \scriptsize (g) $\Gamma_{\mathrm{ph}}: E_{2g}$  B1 \\ \tiny $\Gamma_{\mathrm{mag}}: E_{2u}$; $D_{2h}(C_{2v})$ (AFM-$y$)
\end{minipage}
\begin{minipage}{0.48\columnwidth}
  \centering \includegraphics[width=\textwidth]{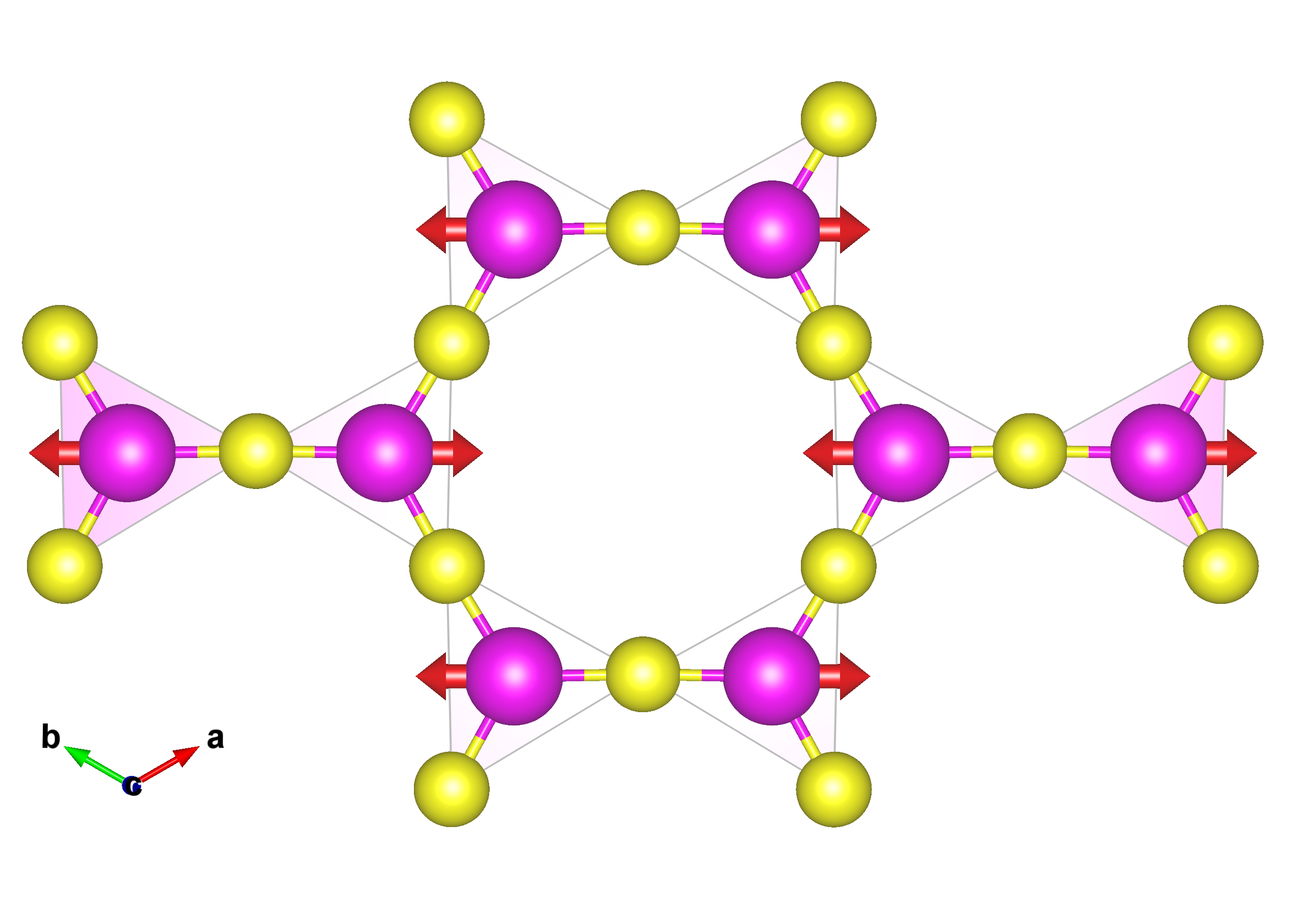} \\ 
  \scriptsize (h) $\Gamma_{\mathrm{ph}}: E_{2g}$ B2 \\ \tiny $\Gamma_{\mathrm{mag}}: E_{2u}$; $D_{2h}(D_{2})$ (AFM-$x$) 
\end{minipage}
\caption{Atlas of mode-resolved symmetry-adapted magnetic templates generated from the eight symmetry-adapted basis realizations associated with the three optical representation sectors ($2E_{1u}\oplus E_{2u}\oplus E_{2g}$) in monolayer $\mathrm{Cd}_2\mathrm{N}_3$. (a--d) The two optical $E_{1u}$ copies ($2E_{1u}$) yield four linearly independent magnetic templates corresponding to the high-symmetry order-parameter directions (Paths B1 and B2) of Opt-1 (a, b) and Opt-2 (c, d), which together span the four-dimensional $2E_{1g}$ in-plane magnetic configuration space with residual $D_{2h}(C_{2h})$ symmetry. (e, f) The two components (Paths B1 and B2) of the optical $E_{2u}$ mode map to staggered out-of-plane AFM-$z$ templates with $D_{2h}(C_{2h})$ and $D_{2h}$ symmetries. (g, h) The two components (Paths B1 and B2) of the optical $E_{2g}$ mode map to collinear in-plane AFM-$y$ and AFM-$x$ templates on Cd sites with $D_{2h}(C_{2v})$ and $D_{2h}(D_{2})$ symmetries.}
\label{fig:2D_manifolds_vertical}
\end{figure}

\subsection{Two-Dimensional Modes: Single and Multi-Copy Representation Spaces}
\label{subsec:2d_modes}

For multi-dimensional irreps ($d_\Gamma > 1$) and/or multiple occurrences of the same irrep ($n_\Gamma > 1$), representation theory defines invariant subspaces within which order-parameter directions (OPDs) select specific broken-symmetry branches (epikernels) [Secs.~\ref{subsec:decoupling} and \ref{subsec:branch_selection}]:

\textit{1. Single Two-Dimensional Optical Modes ($E_{2g}$ and $E_{2u}$ with $d=2, n=1$):---} 
\begin{itemize}
    \item \textbf{$E_{2g}$ Driver $\to E_{2u}$ Magnetic Sector:} Maps directly to collinear in-plane AFM-$xy$ configurations localized exclusively on Cd sites [Fig.~\ref{fig:2D_manifolds_vertical}(g) for Path B1 ($m'mm$, AFM-$y$) and Fig.~\ref{fig:2D_manifolds_vertical}(h) for Path B2 ($m'm'm'$, AFM-$x$)].
    \item \textbf{$E_{2u}$ Driver $\to E_{2g}$ Magnetic Sector:} Maps to staggered out-of-plane AFM-$z$ order localized on N sites [Fig.~\ref{fig:2D_manifolds_vertical}(e) for Path B1 ($mm'm'$) and Fig.~\ref{fig:2D_manifolds_vertical}(f) for Path B2 ($mmm^*$)].
\end{itemize}

\textit{2. Multi-Copy Two-Dimensional Sectors ($3E_{1u} = E_{1u}^{\text{aco}} \oplus 2E_{1u}^{\text{opt}}$ with $d=2, n=3$):---}
\begin{itemize}
    \item \textbf{Acoustic $E_{1u} \to E_{1g}$ ($n=1$):} Provides the two-dimensional in-plane FM-$xy$ template space [Fig.~\ref{fig:stability_multipoles}(c, d)]. Here, relative spin directions are locked, while the overall magnetization direction undergoes spontaneous symmetry breaking along high-symmetry order-parameter directions.
    \item \textbf{Two Optical $E_{1u}$ Copies ($2E_{1u}$) $\to 2E_{1g}$ ($d=2, n=2$):} Because both Cd and N atoms carry finite magnetic moments, the two optical $E_{1u}$ copies map onto the four-dimensional $2E_{1g}$ magnetic representation sector ($\dim(2E_{1g}) = 2 \times 2 = 4$). In accordance with the projection equivalence [Eq.~\eqref{eq:isotypic_iso}], the four configurations shown in Fig.~\ref{fig:2D_manifolds_vertical}(a--d), corresponding to the two symmetry directions of each optical copy, provide four linearly independent, mode-resolved physical realizations that span this four-dimensional projected subspace ($\mathrm{Im}(P_{\mathrm{mag}}^{(2E_{1g})}) = \mathrm{Im}(P_{\mathrm{ph}}^{(2E_{1u})})$).
    
    In general, an arbitrary magnetic order parameter within this representation sector is expressed as a linear superposition:
    \begin{equation}
        \mathbf{M} = \sum_{a=1}^{2} \sum_{\mu=1}^{2} \eta_{a\mu} \mathbf{M}_{a\mu},
    \end{equation}
    where $a=1,2$ indexes the two optical $E_{1u}$ copies and $\mu=1,2$ denotes the two-dimensional irrep directions (Path B1 and Path B2). The four templates in Fig.~\ref{fig:2D_manifolds_vertical}(a--d) correspond to the canonical coordinate directions $(1,0,0,0)$, $(0,1,0,0)$, $(0,0,1,0)$, and $(0,0,0,1)$, establishing the complete kinematic basis within which magnetic energetics perform thermodynamic free-energy minimization and branch selection (\textit{Selection}).
\end{itemize}

% =================================================================
% --- Section VI: Discussion ---
% =================================================================
\section{Discussion}
\label{sec:discussion}

\subsection{Spatial Mapping and Parent-Group Framework}
\label{subsec:spatial_mapping_disc}

The principal theoretical advance of this work is the formulation of a universal determinant-induced pseudoscalar twist acting involutively on the irreducible-representation landscape of crystallographic point groups, and its elevation to an exact representation-to-geometry correspondence. The physical significance of this construction arises because the twist is not merely an abstract relabeling of irreducible representations: in the common site-resolved Cartesian realization, it yields component-by-component identical matrix-element projection operators ($P_{\mathrm{mag}, mn}^{(\Gamma \otimes \Gamma_{\mathrm{ps}})} = P_{\mathrm{ph}, mn}^{(\Gamma)}$) and identical projected configuration subspaces ($\mathrm{Im}(P_{\mathrm{mag}}^{(\Gamma \otimes \Gamma_{\mathrm{ps}})}) = \mathrm{Im}(P_{\mathrm{ph}}^{(\Gamma)})$). This establishes the \textit{Template Principle}, proving that polar phonon modes and axial magnetic configurations are two physical realizations of the same parent-group representation geometry under the pseudoscalar twist, with all symmetry-enforced nodal manifolds strictly preserved.

This geometric equivalence clarifies the relationship between the present framework and conventional magnetic symmetry methods. The symmetry-allowed magnetic configuration space $V_{\mathrm{mag}} = \bigoplus_\Gamma V_\Gamma$ is already rigorously classified by Bertaut's magnetic representation analysis \cite{Bertaut1968}. Our central contribution is not a re-derivation of this known possibility space, but its concrete \textit{parent-lattice physical realization}: while Bertaut's approach establishes which magnetic configurations are symmetry-allowed by constructing abstract basis functions, our framework identifies where their real-space spatial templates already reside within the pre-existing parent crystal prior to energetic selection. In short, representation theory defines the magnetic possibility space, while our pseudoscalar mapping provides its mode-resolved parent-lattice realization within the corresponding magnetic configuration space, with the magnetic-sublattice restriction treated separately through $Q_{\mathrm{mag}}$. The framework is therefore fundamentally complementary to conventional representation and isotropy analyses \cite{Bertaut1968, Stokes1988_Isotropy}: parent-lattice kinematics organizes candidate templates (\textit{Possibility}), thermodynamic free-energy minimization selects the realized order-parameter direction (\textit{Selection}), and emergent magnetic groups capture the residual symmetry of that selected state.

\subsection{General Parent-Group Criterion for Linear Static Magnetic Responses}
\label{subsec:response_framework}

Beyond magnetic structure generation, a primary question is which macroscopic functional responses can couple to a given symmetry-adapted magnetic order parameter. In the conventional magnetic-point-group formulation, response tensors are typically analyzed from the symmetry of an already ordered phase \cite{Landau_StatPhys1}. In our formulation, the spatial-temporal separation principle within the paramagnetic gray parent group $\mathcal{G}_{\mathrm{P}} = G \times \Theta_{\mathcal{T}}$ provides an \emph{a priori} parent-group criterion for identifying symmetry-allowed response channels directly from mode-resolved representations, evaluated before the thermodynamic order-parameter direction is selected.

Within the paramagnetic parent group $\mathcal{G}_{\mathrm{P}} = G \times \Theta_{\mathcal{T}}$, the thermodynamic potential 
\begin{equation}
\varPhi_{\mathrm{resp}}^{(L)} = -\mathcal{R}_{i_1 \dots i_n}(L) \prod_a X_{a, i_a}
\end{equation}
describing the linear coupling between a condensed magnetic order parameter $L \sim (\Gamma_L, \tau_{\mathcal{T}}(L)=-1)$ and external static fields $X_a \sim (\Gamma_{X_a}, \tau_{\mathcal{T}}(X_a))$ must be invariant under the full parent group $\mathcal{G}_{\mathrm{P}}$, which, under the present spatial-temporal separation principle, requires simultaneous invariance under the spatial point group $G$ and the time-reversal group $\Theta_{\mathcal{T}}$. Here, ``linear response'' refers to the response being linear in the external driving fields, while the mode-resolved tensor $\mathcal{R}(L)$ is evaluated to leading order in the condensed magnetic order parameter. A nonzero coupling coefficient is symmetry-allowed if and only if both the spatial point-group invariance and the time-reversal invariance conditions are simultaneously satisfied:
\begin{equation}
\boxed{
\begin{cases}
\Gamma_L \otimes \displaystyle\bigotimes_a \Gamma_{X_a} \supset \Gamma_1, & (G\text{-invariance}), \\[10pt]
\tau_{\mathcal{T}}(L) \displaystyle\prod_a \tau_{\mathcal{T}}(X_a) = +1, & (\Theta_{\mathcal{T}}\text{-invariance}),
\end{cases}
}
\label{eq:dual_screening}
\end{equation}
where $\tau_{\mathcal{T}}(L) = -1$ reflects the time-odd character of the magnetic order parameter. Equation~\eqref{eq:dual_screening} operates directly within the parent group without requiring post-facto tensor expansions under lower-symmetry magnetic subgroups.

Two flagship linear static response channels illustrate this parent-group screening protocol:
\begin{enumerate}[(i)]
    \item \textbf{Linear Magnetoelectric Coupling (2nd-rank):} Described by the potential $\varPhi_{\mathrm{me}}^{(L)} = -\alpha_{jk}(L) E_j H_k$, where $E_j \sim \Gamma_{\mathrm{pol}}$ ($\tau_{\mathcal{T}}(E)=+1$) is the polar electric field and $H_k \sim \Gamma_{\mathrm{ps}} \otimes \Gamma_{\mathrm{pol}}$ ($\tau_{\mathcal{T}}(H)=-1$) is the axial magnetic field. For an order parameter originating from a phonon sector $\Gamma$ ($\Gamma_L = \Gamma \otimes \Gamma_{\mathrm{ps}}$), the system admits a symmetry-allowed magnetoelectric coupling if and only if
    \begin{equation}
    \begin{cases}
    (\Gamma \otimes \Gamma_{\mathrm{ps}}) \otimes \Gamma_{\mathrm{pol}} \otimes \left( \Gamma_{\mathrm{ps}} \otimes \Gamma_{\mathrm{pol}} \right) \supset \Gamma_1, \\[6pt]
    \tau_{\mathcal{T}}(L)\cdot\tau_{\mathcal{T}}(E)\cdot\tau_{\mathcal{T}}(H) = (-1)\cdot(+1)\cdot(-1) = +1.
    \end{cases}
    \end{equation}
    Thus the complete $EHL$ coupling is a scalar under the time-reversal group $\Theta_{\mathcal{T}}$.
    \item \textbf{Linear Piezomagnetic Coupling (3rd-rank):} Described by the potential $\varPhi_{\mathrm{pm}}^{(L)} = -\lambda_{i,kl}(L) H_i \sigma_{kl}$, where $\sigma_{kl} \sim [\Gamma_{\mathrm{pol}}^{\otimes 2}]_{\mathrm{sym}}$ ($\tau_{\mathcal{T}}(\sigma)=+1$) is the symmetric mechanical stress tensor. The system admits a symmetry-allowed piezomagnetic coupling if and only if
    \begin{equation}
    \begin{cases}
    (\Gamma \otimes \Gamma_{\mathrm{ps}}) \otimes \left( \Gamma_{\mathrm{ps}} \otimes \Gamma_{\mathrm{pol}} \right) \otimes [\Gamma_{\mathrm{pol}}^{\otimes 2}]_{\mathrm{sym}} \supset \Gamma_1, \\[6pt]
    \tau_{\mathcal{T}}(L)\cdot\tau_{\mathcal{T}}(H)\cdot\tau_{\mathcal{T}}(\sigma) = (-1)\cdot(-1)\cdot(+1) = +1.
    \end{cases}
    \end{equation}
    The corresponding $LH\sigma$ coupling is likewise time-reversal even.
\end{enumerate}

Crucially, because both the magnetic order parameter ($\Gamma_L = \Gamma \otimes \Gamma_{\mathrm{ps}}$) and the magnetic-field factor ($\Gamma_H = \Gamma_{\mathrm{ps}} \otimes \Gamma_{\mathrm{pol}}$) carry the same pseudoscalar twist, the spatial axial character cancels pairwise at the level of the parent-group invariant ($\Gamma_{\mathrm{ps}}^{\otimes 2} = \Gamma_1$). The resulting spatial selection rules for both magnetoelectric and piezomagnetic couplings reduce directly to ordinary polar tensor products ($\Gamma \otimes \Gamma_{\mathrm{pol}}^{\otimes 2} \supset \Gamma_1$ and $\Gamma \otimes \Gamma_{\mathrm{pol}} \otimes [\Gamma_{\mathrm{pol}}^{\otimes 2}]_{\mathrm{sym}} \supset \Gamma_1$). The response criterion demonstrates that the same spatial-temporal organization extends naturally from magnetic configuration spaces to macroscopic tensorial couplings, providing the parent-group algebraic foundation for identifying functional response channels prior to order-parameter selection. The present work therefore establishes the parent-group kinematic and algebraic foundation ($\text{Possibility} \to \text{Selection} \to \text{Response}$), while companion studies address its constitutive realization in specific response channels.

\subsection{Scope and Limitations}
\label{subsec:scope_limitations}

The pseudoscalar spatial mapping itself is defined for all 32 crystallographic point groups, whereas the present framework is formulated for commensurate $\Gamma$-point ($\mathbf{k}=0$) magnetic orders preserving the parent crystallographic unit cell. It does not currently treat incommensurate structures or spin spirals with finite propagation vectors $\mathbf{q}$. Furthermore, the present formulation does not attempt to resolve the microscopic electronic degrees of freedom of itinerant magnetism; its magnetic configuration space is formulated in terms of classical axial-vector order parameters. The present framework is configurational and symmetry-based on the magnetic side rather than dynamical; it does not by itself determine microscopic exchange couplings, spin-wave dispersions, or quantum many-body states. As throughout the framework, symmetry determines the admissible configuration and coupling space, but does not determine energetic coefficients, transition temperatures, or the thermodynamically selected ground state. Looking forward, because the pseudoscalar twist acts purely on the spatial sector, extending the framework to finite-wavevector representation sectors $\mathbf{k}$ provides a natural direction for future extensions to finite-wavevector magnetic orders. Because the pseudoscalar twist is a structural property of the parent-group representation space itself, the same representation-level organization may provide a broader starting point for symmetry-sensitive structural and functional phenomena beyond the magnetic responses considered here.

% =================================================================
% --- Section VII: Conclusion ---
% =================================================================
\section{Conclusion}
\label{sec:conclusion}

In summary, this work formulates a universal determinant-induced pseudoscalar twist that acts involutively on the irreducible-representation landscape of crystallographic point groups, establishing an exact representation-to-geometry correspondence between parent-lattice vibrations and symmetry-adapted magnetic structures. Within the paramagnetic gray parent group $\mathcal{G}_{\mathrm{P}} = G \times \Theta_{\mathcal{T}}$, the framework implements a factorized spatial-temporal separation principle: the universal pseudoscalar twist governs the spatial polar-to-axial geometric conversion ($\Gamma_{\mathrm{mag}} = \Gamma \otimes \Gamma_{\mathrm{ps}}$), while time-reversal oddness ($\tau_{\mathcal{T}} = -1$) independently defines the magnetic character. Through the component-by-component identity of spatial projection operators ($P_{\mathrm{mag}, mn}^{(\Gamma \otimes \Gamma_{\mathrm{ps}})} = P_{\mathrm{ph}, mn}^{(\Gamma)}$) and identical projected configuration subspaces ($\mathrm{Im}(P_{\mathrm{mag}}^{(\Gamma\otimes\Gamma_{\mathrm{ps}})}) = \mathrm{Im}(P_{\mathrm{ph}}^{(\Gamma)})$) under the common Cartesian realization, we establish the \textit{Template Principle}: parent-phase vibrational normal modes furnish concrete, symmetry-defined real-space templates whose symmetry-enforced nodal manifolds are strictly inherited by magnetic order parameters. The theory does not redefine crystallographic irreps; rather, it exposes an intrinsic, universal involutive structure within their full landscape and provides an exact physical realization bridging abstract magnetic representation spaces to the pre-existing parent crystal.

Application to monolayer $\mathrm{Cd}_2\mathrm{N}_3$ demonstrates how the framework systematically identifies the symmetry channel and real-space configuration of the out-of-plane ferrimagnetic ground state originating from the $A_{2u}$ parent phonon sector (confirmed by first-principles calculations), alongside cluster magnetic octupoles ($B_{1g}, B_{2g}$) and multidimensional antiferromagnetic configuration sectors directly from parent kinematics. The framework thereby organizes the progression from symmetry-defined \textit{Possibility} to thermodynamic \textit{Selection} and, ultimately, to functional \textit{Response}. Because the determinant-induced pseudoscalar twist is an intrinsic structural property of the parent-group representation landscape rather than a material-specific construction, the resulting spatial-temporal architecture provides a general algebraic foundation for companion studies of linear magnetoelectricity and piezomagnetism, as well as broader symmetry-driven phenomena in crystalline materials.

% =================================================================
% --- Section VIII: Acknowledgments ---
% =================================================================

\begin{acknowledgments}
This work was supported by the Advanced Materials-National Science and Technology Major Project of China (Grant No. 2025ZD0620000) and the National Natural Science Foundation of China (Grant Nos. 62174136 and 62674192). H. B. Niu acknowledges financial support from the Natural Science Foundation of Shaanxi Province, China (Grant No. 2021JM-541) and the Youth Innovation Team of Shaanxi Universities.
\end{acknowledgments}

% =================================================================
% --- Section IX: Appendix ---
% =================================================================

\appendix

% 1. 重置并重定义公式编号 (A1, A2, ...)
\setcounter{equation}{0}
\renewcommand{\theequation}{A\arabic{equation}}

% 2. 重设并重定义表格编号 (Table A1, Table A2, ...)
\setcounter{table}{0}
\renewcommand{\thetable}{A\arabic{table}}

% 3. 重设并重定义图片编号 (Fig. A1, Fig. A2, ...) [可选]
\setcounter{figure}{0}
\renewcommand{\thefigure}{A\arabic{figure}}

\section{Representation-Theoretic and Categorical Proof of the Pseudoscalar Mapping}
\hypertarget{app:ctopr}{} 
\label{app:ctopr}

The categorical formulation below is not an additional physical assumption but provides the rigorous structural rationale for the representation-theoretic results established in the main text. To establish the universality and structure-preserving nature of the pseudoscalar mapping across all parent point groups, we provide a representation-theoretic and categorical formulation unifying representation multiplicities, projection operators, and real-space template equivalences.

\subsection{Representation Category and Pseudoscalar Autoequivalence}
\label{subsec:rep_category}

Let $G$ be a finite crystallographic point group, and let $\mathcal{C}=\mathrm{Rep}(G)$ denote the category of its finite-dimensional complex representations. For each $g\in G$, let $R_g\in O(3)$ denote the orthogonal matrix acting on real space. We define the \emph{pseudoscalar representation} $\Gamma_{\mathrm{ps}}$ by the one-dimensional character $\chi_{\mathrm{ps}}(g)=\det(R_g)\in\{+1,-1\}$.

Because $\Gamma_{\mathrm{ps}}\otimes\Gamma_{\mathrm{ps}} \cong \mathbf{1}$ (where $\mathbf{1} \equiv \Gamma_1$ denotes the trivial representation of $G$), $\Gamma_{\mathrm{ps}}$ is an invertible object of order two in $\mathrm{Rep}(G)$. Tensoring by $\Gamma_{\mathrm{ps}}$ defines an autoequivalence functor of the representation category $\mathcal{C} = \mathrm{Rep}(G)$:
\begin{equation}
    \mathcal{F}_{\mathrm{ps}}: \mathcal{C}\rightarrow\mathcal{C}, \quad \mathcal{F}_{\mathrm{ps}}(V) = V\otimes\Gamma_{\mathrm{ps}},
\end{equation}
satisfying the involutive identity $\mathcal{F}_{\mathrm{ps}}^2 \cong \mathrm{Id}_{\mathcal C}$. 

This autoequivalence provides the rigorous representation-theoretic foundation for the physical mappings established in Secs.~\ref{sec:pseudoscalar_mapping} and \ref{sec:rep_analysis}. Specifically, applying $\mathcal{F}_{\mathrm{ps}}$ to the polar vector representation $\Gamma_{\mathrm{polar}}$ [Eq.~\eqref{eq:vec_dual}], to an individual vibrational irreducible representation $\Gamma$ [Eq.~\eqref{eq:irrep_map}], and to the complete parent-lattice vibrational representation $\Gamma_{\mathrm{ph}}^{\mathrm{tot}}$ [Eq.~\eqref{eq:tot_mag_map}] yields their corresponding axial partners:
\begin{align}
    \mathcal{F}_{\mathrm{ps}}(\Gamma_{\mathrm{polar}}) &= \Gamma_{\mathrm{polar}} \otimes \Gamma_{\mathrm{ps}} = \Gamma_{\mathrm{axial}}, \\
    \mathcal{F}_{\mathrm{ps}}(\Gamma) &= \Gamma \otimes \Gamma_{\mathrm{ps}} = \Gamma_{\mathrm{mag}}, \\
    \mathcal{F}_{\mathrm{ps}}(\Gamma_{\mathrm{ph}}^{\mathrm{tot}}) &= \Gamma_{\mathrm{ph}}^{\mathrm{tot}} \otimes \Gamma_{\mathrm{ps}} = \Gamma_{\mathrm{mag}}^{\mathrm{tot}}.
\end{align}
Because $\mathcal{F}_{\mathrm{ps}}$ is an autoequivalence of categories, it is fully faithful and exact, preserving representation dimensions, irreducible decomposition multiplicities, and direct sums while twisting only the spatial parity sector.

\subsection{Multiplicity Preservation, Projection Equivalence, and Nodal Invariance}
\label{subsec:projection_nodal_app}

\begin{proposition}[Multiplicity Preservation via Categorical Autoequivalence]
Let $\mathcal{F}_{\mathrm{ps}}: \mathrm{Rep}(G) \to \mathrm{Rep}(G)$ be the pseudoscalar autoequivalence functor, mapping a generic polar representation space $V_{\mathrm{pol}}$ (corresponding to $\Gamma_{\mathrm{ph}}^{\mathrm{tot}}$ in crystal applications) to its axial partner $V_{\mathrm{ax}} = \mathcal{F}_{\mathrm{ps}}(V_{\mathrm{pol}})$ (corresponding to $\Gamma_{\mathrm{mag}}^{\mathrm{tot}}$). For any irreducible representation $\Gamma$ of $G$, the multiplicity of $\Gamma \otimes \Gamma_{\mathrm{ps}}$ in $V_{\mathrm{ax}}$ equals the multiplicity of $\Gamma$ in $V_{\mathrm{pol}}$:
\begin{equation}
    n_{\mathrm{mag}}^{(\Gamma \otimes \Gamma_{\mathrm{ps}})} = n_{\mathrm{ph}}^{(\Gamma)}.
\end{equation}
\end{proposition}

\begin{proof}
In representation theory, the multiplicity of an irrep $\Gamma$ in a representation $V$ is given by the dimension of the vector space of intertwiners (morphisms in $\mathcal{C}$):
\[
    n_{\mathrm{ph}}^{(\Gamma)} = \dim_{\mathbb{C}} \mathrm{Hom}_G(\Gamma, V_{\mathrm{pol}}).
\]
Because $\mathcal{F}_{\mathrm{ps}}$ is an equivalence of categories (hence fully faithful), it induces a natural isomorphism between Hom-sets:
\[
    \mathrm{Hom}_G(\Gamma, V_{\mathrm{pol}}) \cong \mathrm{Hom}_G\left(\mathcal{F}_{\mathrm{ps}}(\Gamma), \mathcal{F}_{\mathrm{ps}}(V_{\mathrm{pol}})\right) = \mathrm{Hom}_G(\Gamma \otimes \Gamma_{\mathrm{ps}}, V_{\mathrm{ax}}).
\]
Taking vector-space dimensions yields $n_{\mathrm{mag}}^{(\Gamma \otimes \Gamma_{\mathrm{ps}})} = \dim_{\mathbb{C}} \mathrm{Hom}_G(\Gamma \otimes \Gamma_{\mathrm{ps}}, V_{\mathrm{ax}}) = n_{\mathrm{ph}}^{(\Gamma)}$.
\end{proof}

\begin{proposition}[Matrix-Element Projection Equivalence under Common Coefficient Space]
Under the natural identification of the common site-resolved Cartesian coefficient space, the generic spatial matrix-element projection operator $P_{\mathrm{ax}, mn}^{(\Gamma \otimes \Gamma_{\mathrm{ps}})}$ acting on $V_{\mathrm{ax}}$ (denoted $P_{\mathrm{mag}, mn}^{(\Gamma \otimes \Gamma_{\mathrm{ps}})}$ in the main text) is component-by-component identical to $P_{\mathrm{pol}, mn}^{(\Gamma)}$ acting on $V_{\mathrm{pol}}$ (denoted $P_{\mathrm{ph}, mn}^{(\Gamma)}$ in the main text):
\begin{equation}
    P_{\mathrm{ax}, mn}^{(\Gamma \otimes \Gamma_{\mathrm{ps}})} = P_{\mathrm{pol}, mn}^{(\Gamma)}.
\end{equation}
\end{proposition}

\begin{proof}
The matrix-element projection operator for the polar representation $V_{\mathrm{pol}}$ onto the $mn$-th component of $\Gamma$ is defined as
\[
    P_{\mathrm{pol}, mn}^{(\Gamma)} = \frac{d_\Gamma}{|G|} \sum_{g\in G} D_{mn}^{(\Gamma)}(g)^* \hat{D}_{\mathrm{pol}}(g).
\]
For the twisted axial realization $V_{\mathrm{ax}} = V_{\mathrm{pol}} \otimes \Gamma_{\mathrm{ps}}$, the representation matrix element becomes $D_{mn}^{(\Gamma \otimes \Gamma_{\mathrm{ps}})}(g) = D_{mn}^{(\Gamma)}(g) \chi_{\mathrm{ps}}(g)$, and the spatial action is $\hat{D}_{\mathrm{ax}}(g) = \chi_{\mathrm{ps}}(g) \hat{D}_{\mathrm{pol}}(g).$ Substituting these expressions into the projection operator gives
\[
    P_{\mathrm{ax}, mn}^{(\Gamma\otimes\Gamma_{\mathrm{ps}})} = \frac{d_{\Gamma}}{|G|} \sum_{g\in G} \left[ D_{mn}^{(\Gamma)}(g) \chi_{\mathrm{ps}}(g) \right]^* \left[ \chi_{\mathrm{ps}}(g) \hat{D}_{\mathrm{pol}}(g) \right].
\]
Since $\chi_{\mathrm{ps}}(g) = \det(R_g) \in \{+1, -1\}$ is real and satisfies $[\chi_{\mathrm{ps}}(g)]^2 = 1$ for all orthogonal operations $g \in G$, the parity factors cancel identically at every group element:
\[
    P_{\mathrm{ax}, mn}^{(\Gamma\otimes\Gamma_{\mathrm{ps}})} = \frac{d_{\Gamma}}{|G|} \sum_{g\in G} D_{mn}^{(\Gamma)}(g)^* [\chi_{\mathrm{ps}}(g)]^2 \hat{D}_{\mathrm{pol}}(g) = P_{\mathrm{pol}, mn}^{(\Gamma)}.
\]
Crucially, this projection equivalence does not follow from the abstract category equivalence of $\mathrm{Rep}(G)$ alone, but arises from combining the representation parity twist with the common underlying Cartesian coefficient-space realization, leading to the exact cancellation of pseudoscalar characters $[\chi_{\mathrm{ps}}(g)]^2 \equiv 1$.
\end{proof}

\begin{corollary}[Character Projection Equivalence, Isotypic Natural Isomorphism, and Nodal Invariance]
Summing the matrix-element projection operators over diagonal components $m=n$ yields the character projection identity $P_{\mathrm{ax}}^{(\Gamma \otimes \Gamma_{\mathrm{ps}})} = P_{\mathrm{pol}}^{(\Gamma)}$. Consequently, the isotypic component space $\mathrm{Im}(P_{\mathrm{ax}}^{(\Gamma \otimes \Gamma_{\mathrm{ps}})})$ is identical to $\mathrm{Im}(P_{\mathrm{pol}}^{(\Gamma)})$ under the common Cartesian coefficient-space identification, and naturally isomorphic as an abstract representation space. 

Furthermore, let $\mathbf{\Psi}_{\mathrm{pol}}^{(\Gamma)}(\mathbf{r}_0) \in \mathbb{C}^{3 \times d_\Gamma}$ denote the site-resolved Cartesian coefficient matrix whose columns are the basis functions of the polar irrep $\Gamma$ evaluated at position $\mathbf{r}_0$. For any operation $h \in H_{\mathbf{r}_0}$ in the local stabilizer subgroup preserving position $\mathbf{r}_0$ ($h\mathbf{r}_0 = \mathbf{r}_0$), covariance under the spatial action requires
\begin{equation}
    R_h \mathbf{\Psi}_{\mathrm{pol}}^{(\Gamma)}(\mathbf{r}_0) = \mathbf{\Psi}_{\mathrm{pol}}^{(\Gamma)}(\mathbf{r}_0) D^{(\Gamma)}(h).
    \label{eq:nodal_pol}
\end{equation}
For the corresponding axial basis matrix $\mathbf{\Psi}_{\mathrm{ax}}^{(\Gamma \otimes \Gamma_{\mathrm{ps}})}(\mathbf{r}_0)$, transforming under the spatial axial action $\det(R_h) R_h$ and representation matrix $D^{(\Gamma \otimes \Gamma_{\mathrm{ps}})}(h) = \det(R_h) D^{(\Gamma)}(h)$, the covariance condition reads
\begin{equation}
    \det(R_h) R_h \mathbf{\Psi}_{\mathrm{ax}}^{(\Gamma \otimes \Gamma_{\mathrm{ps}})}(\mathbf{r}_0) = \mathbf{\Psi}_{\mathrm{ax}}^{(\Gamma \otimes \Gamma_{\mathrm{ps}})}(\mathbf{r}_0) \left[ \det(R_h) D^{(\Gamma)}(h) \right].
\end{equation}
Because $\det(R_h) = \pm 1 \neq 0$, the scalar determinant cancels identically from both sides, yielding the exact same covariance constraint on the Cartesian coefficient matrix:
\begin{equation}
    R_h \mathbf{\Psi}(\mathbf{r}_0) = \mathbf{\Psi}(\mathbf{r}_0) D^{(\Gamma)}(h).
    \label{eq:nodal_cond}
\end{equation}
Consequently, all symmetry-enforced vanishing components and nodal constraints imposed by the local site stabilizer are identical between polar and axial realizations, proving that all symmetry-enforced nodal manifolds are strictly invariant under the pseudoscalar autoequivalence functor $\mathcal{F}_{\mathrm{ps}}$.
\end{corollary}

\subsection{Compatibility of the Pseudoscalar Twist with Gray-Group Corepresentations}
\label{subsec:corep_compat}

The parent paramagnetic phase before magnetic ordering is formally described by the gray magnetic parent group $\mathcal{G}_{\mathrm{P}} = G \times \Theta_{\mathcal{T}} = G \cup \mathcal{T}G$, where $G$ is the unitary spatial crystallographic point group and $\Theta_{\mathcal{T}} = \{E, \mathcal{T}\}$ is the time-reversal group. In the classical macroscopic setting considered in this work, physical vector fields and order parameters are classified by their decoupled spatial irreducible representation $\Gamma$ and time-reversal parity $\tau_{\mathcal{T}} \in \{+1, -1\}$. Parent-phase vibrational normal modes transform as polar, time-even displacement fields $(\Gamma_{\mathrm{ph}}, +1)$, whereas symmetry-adapted magnetic order parameters transform as axial, time-odd moment configurations $(\Gamma_{\mathrm{mag}}, -1) = (\Gamma_{\mathrm{ph}} \otimes \Gamma_{\mathrm{ps}}, -1)$.

At the full gray-group level, this complete polar-to-magnetic transformation requires the simultaneous twisting of spatial parity and inversion of time-reversal grading. In the present real classical order-parameter setting, this combined mapping is formalized as a factorized one-dimensional antiunitary twist, symbolically denoted by:
\begin{equation}
    \widetilde{\Gamma}_{\mathrm{ps}}^- = \Gamma_{\mathrm{ps}} \boxtimes \mathbf{1}^-,
\end{equation}
where $\Gamma_{\mathrm{ps}}$ acts on the spatial sector via $\chi_{\mathrm{ps}}(g) = \det(R_g)$ for all $g \in G$, and $\mathbf{1}^-$ specifies the time-reversal-odd grading $\tau_{\mathcal{T}} = -1$. 

Rather than invoking formal tensor products on antiunitary representation categories, this transformation is explicitly defined as an algebraic twist on the gray-group corepresentation $\mathfrak{D} = \{D(g), D(g\mathcal{T})\}$:
\begin{equation}
    \widetilde{\mathcal{F}}_{\mathrm{ps}}^-: \mathfrak{D} \longmapsto \widetilde{\mathfrak{D}},
\end{equation}
whose matrix actions on the unitary and antiunitary elements are respectively given by:
\begin{align}
    \widetilde{D}(g) &= \chi_{\mathrm{ps}}(g) D(g), \label{eq:corep_twist_u} \\
    \widetilde{D}(g\mathcal{T}) &= -\chi_{\mathrm{ps}}(g) D(g\mathcal{T}), \label{eq:corep_twist_a}
\end{align}
for all $g \in G$. Because $\chi_{\mathrm{ps}}(g) \in \{+1, -1\} \subset \mathbb{R}$ is real, the scalar factor commutes identically with the antilinear complex conjugation associated with the antiunitary operations, ensuring that the twisted operators $\widetilde{\mathfrak{D}}$ obey the corresponding Wigner corepresentation multiplication relations of $\mathcal{G}_{\mathrm{P}}$.

Applied to the parent-phase vibrational normal modes, this twist executes the complete physical mapping from polar lattice vibrations to axial magnetic configurations:
\begin{equation}
    \widetilde{\mathcal{F}}_{\mathrm{ps}}^-: (\Gamma_{\mathrm{ph}}, +1) \longmapsto (\Gamma_{\mathrm{ph}} \otimes \Gamma_{\mathrm{ps}}, -1) = (\Gamma_{\mathrm{mag}}, -1).
\end{equation}

Let $U: \mathrm{Corep}(\mathcal{G}_{\mathrm{P}}) \to \mathrm{Rep}(G)$ denote the restriction functor from gray-group corepresentations to representations of the unitary spatial subgroup $G$. Since the temporal factor in Eq.~\eqref{eq:corep_twist_a} acts exclusively on the antiunitary coset $G\mathcal{T}$, the restriction to the unitary subgroup satisfies the exact functorial compatibility relation:
\begin{equation}
    U \circ \widetilde{\mathcal{F}}_{\mathrm{ps}}^- \simeq \mathcal{F}_{\mathrm{ps}} \circ U,
\end{equation}
making the following diagram commute:
\[
\begin{CD}
\mathrm{Corep}(\mathcal{G}_{\mathrm{P}}) @>\widetilde{\mathcal F}_{\mathrm{ps}}^->> \mathrm{Corep}(\mathcal{G}_{\mathrm{P}})\\
@V U VV @VV U V\\
\mathrm{Rep}(G) @>\mathcal F_{\mathrm{ps}}>> \mathrm{Rep}(G)
\end{CD}
\]

Furthermore, because $[\chi_{\mathrm{ps}}(g)]^2 \equiv 1$ and $(-1)^2 = +1$, applying the twist twice restores the original corepresentation, establishing that the transformation is involutive up to natural equivalence:
\begin{equation}
    (\widetilde{\mathcal{F}}_{\mathrm{ps}}^-)^2 \simeq \mathrm{Id}_{\mathrm{Corep}(\mathcal{G}_{\mathrm{P}})}.
\end{equation}

Crucially, the factorized structure $\widetilde{\Gamma}_{\mathrm{ps}}^- = \Gamma_{\mathrm{ps}} \boxtimes \mathbf{1}^-$ provides the exact representation-theoretic foundation for the spatial-temporal separation principle: the spatial geometric configuration of the magnetic order is governed entirely by the spatial pseudoscalar twist $\mathcal{F}_{\mathrm{ps}}$, while the independent factor $\mathbf{1}^-$ specifies its time-reversal-odd magnetic character. The corepresentation twist is invoked here as a structural compatibility extension for the classical gray parent group; the constructive template derivations, matrix-element identities ($P_{\mathrm{mag}} = P_{\mathrm{ph}}$), and nodal manifold constraints in the main text remain formulated within the unitary spatial representation category.

\subsection{Categorical Interpretation and Physical Dictionary}
\label{subsec:cat_dict}

The categorical terminology in this appendix is not an additional physical layer but exposes the structural reason behind the representation-preserving and involutive properties of the pseudoscalar mapping. It also distinguishes the representation-level twist from the site-restriction and projection operations used for concrete template construction. Table~\ref{tab:dictionary} summarizes the conceptual correspondence between representation-theoretic definitions, category-theoretic terminology, and physical actions.

\begin{table*}[htb]
\centering
\small
\renewcommand{\arraystretch}{1.35}
\caption{Representation--category--physics correspondence within the parent-group framework. The table summarizes the principal representation-theoretic structures, their category-theoretic counterparts, and their explicit mathematical or physical meanings.}
\label{tab:dictionary}
\begin{tabularx}{\textwidth}{@{} >{\raggedright\arraybackslash}p{0.27\textwidth} >{\raggedright\arraybackslash}p{0.27\textwidth} X @{}}
\toprule
\textbf{Representation / Physical Object}
&
\textbf{Category-Theoretic Term}
&
\textbf{Mathematical / Physical Meaning and Constructive Action}
\\
\midrule

\multicolumn{3}{@{}l}{\textit{\textbf{1. Unitary Spatial Representation Category $\mathrm{Rep}(G)$}}} \\

Single vibrational sector $\Gamma$
&
Simple object (irrep) $\Gamma \in \mathrm{Irr}(G)$
&
Independent symmetry-defined spatial sector of the parent lattice
\\

Total parent-lattice vibrational representation space
&
Semisimple object $V \in \mathrm{Rep}(G)$
&
$V_{\mathrm{ph}}^{\mathrm{tot}}
=
\bigoplus_{\Gamma}
n_{\mathrm{ph}}^{(\Gamma)}\Gamma$
\\

Irrep multiplicity $n_{\mathrm{ph}}^{(\Gamma)}$
&
Hom-space dimension
$\dim_{\mathbb C}\mathrm{Hom}_G(\Gamma,V)$
&
Multiplicity of $\Gamma$ in the parent vibrational representation:
$n_{\mathrm{ph}}^{(\Gamma)}
=
\dim_{\mathbb C}\mathrm{Hom}_G
(\Gamma,V_{\mathrm{ph}}^{\mathrm{tot}})$
\\

Symmetry-preserving map
&
Morphism (intertwiner)
$f\in\mathrm{Hom}_G(V,W)$
&
Linear map satisfying
$fD_V(g)=D_W(g)f$
\\

Pseudoscalar representation $\Gamma_{\mathrm{ps}}$
&
Invertible object of order two in $\mathrm{Rep}(G)$
&
Universal one-dimensional spatial parity twist:
$\chi_{\mathrm{ps}}(g)=\det(R_g)$
\\

Pseudoscalar representation twist
&
Autoequivalence functor
$\mathcal{F}_{\mathrm{ps}}:V\mapsto V\otimes\Gamma_{\mathrm{ps}}$
&
Structure-preserving polar--axial representation conversion
\\

Representation-level involution
&
Involutive functor
$\mathcal{F}_{\mathrm{ps}}^2\simeq\mathrm{Id}_{\mathrm{Rep}(G)}$
&
$\Gamma_{\mathrm{ps}}^{\otimes2}\cong\Gamma_1$, hence reversible self-inversion of the representation twist
\\

Single-mode mapping
$\Gamma\mapsto\Gamma\otimes\Gamma_{\mathrm{ps}}$
&
Functor action $\mathcal{F}_{\mathrm{ps}}(\Gamma)$
&
Direct one-to-one assignment of a parent vibrational sector to its symmetry-adapted magnetic sector
\\

Total representation mapping
$\Gamma_{\mathrm{mag}}^{\mathrm{tot}}
=
\Gamma_{\mathrm{ph}}^{\mathrm{tot}}\otimes\Gamma_{\mathrm{ps}}$
&
Functor action
$\mathcal{F}_{\mathrm{ps}}(V_{\mathrm{ph}}^{\mathrm{tot}})$
&
Global mapping of the complete parent-lattice vibrational configuration space to the magnetic configuration space
\\

Multiplicity conservation
$n_{\mathrm{mag}}^{(\Gamma\otimes\Gamma_{\mathrm{ps}})}
=
n_{\mathrm{ph}}^{(\Gamma)}$
&
Induced natural isomorphism of Hom-spaces
&
$n_{\mathrm{mag}}^{(\Gamma\otimes\Gamma_{\mathrm{ps}})} = n_{\mathrm{ph}}^{(\Gamma)}$; exact preservation of irrep multiplicities and isotypic-sector dimensions
\\

\midrule
\multicolumn{3}{@{}l}{\textit{\textbf{2. Real-Space Coordinate Realization and Geometry}}} \\

Common Cartesian coefficient space
&
Real coordinate realization
&
Unified site-resolved Cartesian space for displacement and localized magnetic-moment components
\\

Matrix projector identity
&
Equality of projected subspaces
&
$P_{\mathrm{mag},mn}^{(\Gamma\otimes\Gamma_{\mathrm{ps}})}
=
P_{\mathrm{ph},mn}^{(\Gamma)}$;
this yields the Template Principle
\\

Isotypic-space identity
&
Natural isomorphism of representation spaces
&
$\mathrm{Im}(P_{\mathrm{mag}}^{(\Gamma\otimes\Gamma_{\mathrm{ps}})}) = \mathrm{Im}(P_{\mathrm{ph}}^{(\Gamma)})$ under the natural Cartesian identification
\\

Local site-stabilizer covariance
&
Stabilizer-induced nodal invariance
&
Exact preservation of symmetry-enforced vanishing components and nodal manifolds
\\

\midrule
\multicolumn{3}{@{}l}{\textit{\textbf{3. Paramagnetic Gray Parent Group $\mathcal{G}_{\mathrm{P}}=G\times\Theta_{\mathcal T}$}}} \\

Paramagnetic gray parent group
&
Corepresentation category
$\mathrm{Corep}(\mathcal{G}_{\mathrm{P}})$
&
Full high-temperature parent symmetry combining the spatial group $G$ and time-reversal group $\Theta_{\mathcal T}$
\\

Restriction to the spatial subgroup $G$
&
Restriction functor
$U:\mathrm{Corep}(\mathcal{G}_{\mathrm{P}})\to\mathrm{Rep}(G)$
&
Retention of the unitary spatial transformation structure while discarding the antiunitary action
\\

Time-odd factorized twist
$\widetilde{\Gamma}_{\mathrm{ps}}^-
=
\Gamma_{\mathrm{ps}}\boxtimes\mathbf{1}^-$
&
One-dimensional gray-group twist
&
Combined spatial--temporal mapping
$(\Gamma_{\mathrm{ph}},+1)
\mapsto
(\Gamma_{\mathrm{mag}},-1)$
\\

Functorial compatibility
&
Commutative diagram
$U\circ\widetilde{\mathcal F}_{\mathrm{ps}}^-
\simeq
\mathcal F_{\mathrm{ps}}\circ U$
&
Formal compatibility of the spatial pseudoscalar twist with the independent time-reversal grading
\\

Full magnetic twist involution
&
Involutive equivalence
$(\widetilde{\mathcal F}_{\mathrm{ps}}^-)^2
\simeq
\mathrm{Id}$
&
Algebraically self-inverse mapping on the full gray-group symmetry structure
\\

\bottomrule
\end{tabularx}
\end{table*}

\bibliographystyle{apsrev4-2}
\enlargethispage{\baselineskip}
\bibliography{References}

\end{document}